\documentclass[5p,final]{elsarticle}

\usepackage[utf8]{inputenc}
\usepackage[T1]{fontenc}
\usepackage[english]{babel}
\usepackage{amsmath,amssymb,amsfonts,bm}
\usepackage{amsthm}
\usepackage{graphicx}
\usepackage{booktabs}
\usepackage{placeins}
\usepackage{tikz}
\usetikzlibrary{arrows.meta,calc}
\usepackage{pgfplots}
\pgfplotsset{compat=1.18}
\usepackage[colorlinks=true,linkcolor=blue,citecolor=blue,urlcolor=blue]{hyperref}

\newcommand{\FigDir}{figures}
\graphicspath{{\FigDir/}}

\newcommand{\E}{\mathbb{E}}
\newcommand{\Q}{\mathbb{Q}}
\newcommand{\M}{\mathbb{M}}
\newcommand{\Var}{\operatorname{Var}}
\newcommand{\dd}{\,\mathrm{d}}
\newcommand{\ii}{\mathrm{i}}
\newcommand{\wh}[1]{\widehat{#1}}
\newcommand{\cL}{\mathcal{L}}
\newcommand{\sigloc}{\sigma_{\mathrm{loc}}}
\newcommand{\Aeff}{\mathcal{A}_{\mathrm{eff}}}
\newcommand{\Fcal}{\mathcal{F}}
\newcommand{\Xiop}{\Xi}
\newcommand{\V}{\mathcal{V}}
\newcommand{\Ad}{\operatorname{Ad}}
\newcommand{\Pair}[2]{\left\langle #1,#2\right\rangle}
\newcommand{\Bop}{\mathcal{B}}
\newcommand{\Gop}{\mathcal{G}}
\newcommand{\Pproj}{\mathcal{P}}
\newcommand{\Kern}{\mathcal{K}}
\newcommand{\Sdom}{\mathcal S}
\newcommand{\Dclaim}{\mathcal{D}_{\!B}}
\newcommand{\Ddens}{\mathcal{D}_{\!G}}
\newcommand{\Ical}{\mathcal{I}}

\newcommand{\PairS}[2]{\left\langle #1,#2\right\rangle_{\Sdom}}

\theoremstyle{definition}
\newtheorem{definition}{Definition}
\theoremstyle{plain}
\newtheorem{proposition}{Proposition}
\newtheorem{theorem}{Theorem}
\theoremstyle{remark}
\newtheorem{remark}{Remark}

\begin{document}
\begin{frontmatter}

\title{Reaction-boundary variance and adjoint-consistent local-volatility projection}

\author[unsw-mathstats]{Christopher Angstmann}\ead{c.angstmann@unsw.edu.au}
\address[unsw-mathstats]{School of Mathematics and Statistics, University of New South Wales, Sydney, NSW 2052, Australia}
\author[uct-sta]{Tim Gebbie} \ead{tim.gebbie@uct.ac.za}
\address[uct-sta]{Department of Statistical Sciences, University of Cape Town, Rondebosch 7701, Western Cape, South Africa}
\begin{abstract}
We derive an operational-time variance kernel for a latent-order-book reaction boundary and use it to separate three objects usually collapsed in calendar-time volatility models: a structural boundary cumulant, a clock projection, and a pricing-measure choice. The reaction boundary is the zero of a bid--ask imbalance field. For a locally linear book, signed order-flow perturbations displace this zero through a damped Abel response kernel, so the variance of boundary increments is obtained as a finite-scale Green-function cumulant rather than introduced as a primitive diffusion coefficient. For long-memory forcing with exponent $0<\gamma<1$, the operational variance has a closed asymptotic form involving effective signed-forcing intensity, liquidity slope, resilience, memory, and operational coarse-graining scale. A deterministic activity clock gives the benchmark local-volatility projection when a diffusion closure is admissible. More general, non-unique clocks generate candidate calendar-time pricing systems. We argue that such projections are admissible only when the induced forward density operator and backward valuation operator remain adjoint on the same state space. Adjoint consistency is therefore a reality constraint on operational-to-calendar time projection: it disciplines non-unique time and identifies where clock ambiguity enters.
\end{abstract}


\end{frontmatter}
\section{Introduction}

The motivating financial question is whether the local variance entering a calendar-time option-pricing model can be related to fluctuations of the price interface at which local bid and ask liquidity balance, and, if so, how that variance should be carried from market activity time into calendar time.

Financial markets need not evolve uniformly in calendar time. Trading activity arrives through irregular sequences of orders, trades, quotes and cancellations, so it is often useful to describe price formation first on an activity, or operational, clock and only afterwards relate that evolution to calendar time \cite{Clark1973,CarrWu2004,AngstmannGebbie2026NonUniqueTime}. In the present paper, operational time $u$ indexes the local market activity on which the reaction-boundary dynamics evolve, whereas calendar time $t$ carries the maturity, discounting and contractual timing of the option-pricing problem. A deterministic activity clock $u=U(t)$ provides the simplest link between the two; more general clock choices raise the projection-consistency question considered later in the paper.

At the microstructure level we begin with market objects that generate price changes rather than from a price diffusion itself. A latent order book represents buy- and sell-side liquidity beyond the orders currently displayed in the visible book, and the local signed imbalance between these two sides provides a convenient state variable for price formation \cite{bouchaud2009markets,toth2011anomalous,Mastromatteo2014,Donier2015}. In this representation price is associated with the moving interface at which local bid and ask liquidity balance. The calculation below then asks how fluctuations of this interface generate a finite-scale variance before any calendar-time pricing model is imposed.

Volatility is often introduced at the level of a calendar-time diffusion, either postulated directly as in Black--Scholes--Merton or inferred from option prices in local-volatility models \cite{BlackScholes1973,Merton1973,Dupire1994,Gatheral2006}. The inference viewpoint is also linked to state-price recovery from option prices, where cross-strike curvature carries information about risk-neutral densities \cite{BreedenLitzenberger1978}. Here we reverse this order. We first compute a mesoscopic transport coefficient for a reaction boundary in a latent order book. Here the mesoscopic transport coefficient is the variance-rate coefficient that remains local in state and operational time but is resolved over a non-zero operational window rather than identified with an instantaneous microscopic diffusion coefficient \cite{AngstmannGebbieBSMFromLattice2026}. Only after this operational-time coefficient is obtained do we discuss its projection into calendar time.

The reaction boundary is the log-price location at which a local bid--ask imbalance field changes sign. For a locally linear latent book, small imbalance perturbations are converted into boundary displacements by the inverse local book slope. Signed order-flow perturbations are propagated to the boundary by the same damped Abel response kernel that appears in latent-liquidity accounts of concave market impact \cite{bouchaud2009markets,toth2011anomalous,Mastromatteo2014,Donier2015}. However, the object computed here is not the mean impact path; it is the second cumulant of the boundary displacement over a finite operational scale.

Concretely, let $S(u)$ denote the log-price value of the reaction boundary in operational time $u$. The formal infinitesimal target for the operational variance of the reaction boundary is then:
\begin{equation}
 a_u(S,u)=\lim_{h\downarrow0}\frac{1}{h}
 \E\left[(S(u+h)-S(u))^2\mid S(u)=S\right].
 \label{eq:au-target}
\end{equation}
When this limit exists, it fixes the transport coefficient that would appear in an admissible local diffusion closure. In the Green-function calculation below the Abel kernel is singular at the origin, and the operational variance is therefore realised as a regularised, coarse-grained coefficient $a_u^{(\Delta)}(S,u)$ over a mesoscopic operational scale $\Delta$. The limiting notation in Eq.~\eqref{eq:au-target} should thus be read as the local target of the closure, while the object carried into the spectral calculation and clock projection is the finite-scale operational cumulant.

Here coarse-grained means that the variance is measured over a non-zero operational window $\Delta$, with fluctuations below that scale represented through the finite-scale coefficient. The resulting mesoscopic closure is the finite-scale reduction of the detailed reaction-boundary dynamics to $a_u^{(\Delta)}(S,u)$, rather than an assertion that the microscopic boundary is itself a diffusion in operational or calendar time \cite{AngstmannGebbieBSMFromLattice2026}.

Concretely, for an admissible diffusion closure and a deterministic activity clock $u=U(t)$, the calendar-time local variance used in the mesoscopic closure is obtained by multiplying the operational variance by the activity rate,
\begin{equation}
 \sigloc^2(S,t)=\alpha(t)a_u^{(\Delta)}(S,U(t)),
 \qquad \alpha(t)=\frac{\dd U(t)}{\dd t}.
 \label{eq:clock-map}
\end{equation}
When the scale $\Delta$ is suppressed, Eq.~\eqref{eq:clock-map} reduces to the usual shorthand $\sigloc^2=\alpha a_u$. Thus the local-volatility coefficient is not assumed as a primitive property of the traded price. It is an activity-rescaled boundary cumulant: the order-book response determines the operational variance, while the clock determines how that operational variance is sampled in calendar time.

This construction separates three layers that are often collapsed into a single volatility coefficient. The first layer is the local reaction-boundary response of the latent book. The second layer is the cumulant calculation for long-memory signed forcing in operational time. The third layer is the projection from operational time to calendar time. In this respect the paper is close to the subordination and time-change tradition, but differs by leaving the clock as an admissibility problem rather than specifying it at the outset \cite{Bochner1949,Clark1973,CarrWu2004}. A deterministic clock gives the benchmark local-volatility equation. A non-unique clock leaves the operational cumulant unchanged but changes the projection, so that incompleteness or clock uncertainty enters at the time-change layer rather than inside the Green-function kernel itself \cite{AngstmannGebbie2026NonUniqueTime}.

In summary: the purpose here is to derive a finite-scale operational variance coefficient for a reaction-boundary price process and to identify the conditions under which its projection into calendar time defines a local-volatility pricing model. The construction combines latent-order-book response, long-memory signed forcing, and forward--backward adjoint consistency. The result is a structural decomposition of local volatility into boundary variance, activity, and pricing-measure choice.

\section{Main results and contribution}

The attack on the problem starts with a locally linear reaction boundary: the bid--ask imbalance field has a simple zero and, close to that zero, the latent book is represented by a positive local liquidity slope $\cL_u>0$ \cite{bouchaud2009markets,toth2011anomalous,Mastromatteo2014,Donier2015}. We then treat small boundary displacements to first order over a local response window. Over this calculation the book slope, replenishment/diffusion scale, resilience and boundary position are held fixed, with their state and operational-time dependence restored afterwards \cite{Mastromatteo2014,Donier2015,Benzaquen2018}. This is the local response approximation used to construct the boundary displacement before its variance is calculated (see Sections~\ref{sec:reaction-boundary-response}--\ref{sec:green-function-cumulant}).

The forcing is taken to be locally stationary with long memory, $0<\gamma<1$, and with a finite short-scale cutoff $\tau_0>0$, so the variance is kept at a non-zero operational scale $\Delta$ rather than identified with an unregularised microscopic limit \cite{LilloMikeFarmer2005,Beran1994,Priestley1981}. The exact finite-cutoff calculation is followed by the separated-scale approximation $\Delta/\tau_0\gg1$, with the additional high-resilience zero-cutoff restriction $\nu\tau_0\ll1$ (see Sections~\ref{sec:green-function-cumulant}--\ref{sec:asymptotic-closure} and \ref{app:finite-cutoff}). To pass from this operational variance to the benchmark calendar-time pricing problem we take $u=U(t)$ to be absolutely continuous, with activity rate $\alpha(t)=\dd U(t)/\dd t$, and use the risk-neutral projected kernel in the usual no-arbitrage setting \cite{Clark1973,CarrWu2004,HarrisonPliska1981} (see Section~\ref{sec:clock-projection-pricing}). If the operational-to-calendar projection is instead non-unique, the backward valuation and forward state-price operators must refer to the same state space and information set and satisfy the adjoint consistency condition \cite{EthierKurtz1986,Dynkin1965MarkovProcessesI,Feller1971Volume2} (see Section~\ref{sec:adjoint-real} and \ref{app:adjoint-consistency}).

The paper has three linked results. The first is the operational reaction-boundary variance kernel. Under the local assumptions stated above, the finite-scale operational variance rate satisfies
\begin{equation}
 a_u^{(\Delta)}(S,u)
 \simeq
 \left[\frac{\Aeff(S,u)}{\cL_u(S,u)^2}\right]
 \Delta^{-\gamma(S,u)}
 \Fcal_{\gamma(S,u)}(\nu_u(S,u)\Delta),
 \label{eq:main-operational-result}
\end{equation}
where $\Aeff=A_m/D_u$ and $\Fcal_\gamma$ is the dimensionless response function in Eq.~\eqref{eq:F-definition}. This is a mesoscopic transport coefficient: it is local in $(S,u)$ but resolved at operational scale $\Delta$.

The second result is the deterministic-clock projection. If a diffusion closure is admissible and $u=U(t)$ is absolutely continuous with activity rate $\alpha(t)=\dd U/\dd t$, then the corresponding calendar-time local variance is
\begin{equation}
 \sigloc^2(S,t)=\alpha(t)\Xiop(S,U(t);\Delta),
 \label{eq:main-clock-result}
\end{equation}
where $\Xiop$ denotes the operational kernel in Eq.~\eqref{eq:Xi}. Under a chosen pricing measure this gives the benchmark local-volatility PDE in Eq.~\eqref{eq:pde}.

The third result is a consistency criterion for non-unique time. If the clock is drawn from a family $\mathfrak U$, a projection of $\Xiop$ into calendar time is not automatically a pricing model. It is a pricing model only when the projected backward valuation operator and the projected forward density operator are adjoint on the same state space and information set. This is the adjoint-reality constraint developed in Section~\ref{sec:adjoint-real}. It is the step that distinguishes a coherent calendar-time representation from a formal variance projection.

These three results are then drawn together into the overall contribution, which is to provide a structural route from latent-order-book reaction-boundary fluctuations to an admissible calendar-time local-volatility pricing representation, with the operational variance derived first and non-uniqueness or incompleteness entering through the subsequent clock and pricing-measure projection.

\section{Reaction-boundary response}
\label{sec:reaction-boundary-response}
Let $\Phi(x,u)$ be the bid--ask imbalance field. Here $\Phi(x,u)$ is the signed local difference between bid-side and ask-side latent liquidity at log-price coordinate $x$ and operational time $u$ \cite{bouchaud2009markets,toth2011anomalous,Mastromatteo2014,Donier2015}. The reaction boundary is defined by
\begin{equation}
 \Phi(y(u),u)=0.
 \label{eq:zero-boundary}
\end{equation}
We identify the reaction-boundary coordinate with the operational log-price state, $S(u)=y(u)$, and write $Y(u)=S(u)-S_0=y(u)-y_0$ for its centred displacement. Thus conditioning on $S(u)=S$ is equivalent to conditioning on the centred displacement $Y(u)=S-S_0$.

Near a simple zero, assume a locally linear latent book \cite{bouchaud2009markets,toth2011anomalous,Donier2015},
\begin{equation}
 \Phi^*(x,u)\simeq-\cL_u(x-y(u)),
  \label{eq:linear-book}
 \end{equation}
 where $\cL_u=-\partial_x\Phi^*(y(u),u)>0$.
 
A perturbation $\Psi$ obeys the frozen-coefficient response equation
\begin{equation}
 \partial_u\Psi=D_u\partial_{xx}\Psi-\nu_u\Psi+m(u)\delta(x-y(u)),
 \label{eq:response-pde}
\end{equation}
where $D_u$ is replenishment/diffusion, $\nu_u$ resilience, and $m(u)$ signed forcing. The associated boundary kernel is
\begin{equation}
 g_{\nu,D}(\tau)=\frac{e^{-\nu\tau}}{\sqrt{4\pi D\tau}},\qquad \tau>0.
 \label{eq:abel-kernel}
\end{equation}
This is the ordinary operational-time diffusion branch ($\alpha_u=1$) of the broader DTRW reaction--diffusion framework considered for coupled order books in \cite{AngstmannGebbieEpps2026}.
The boundary displacement is therefore, to first order,
\begin{equation}
 Y(u)=y(u)-y_0\simeq\frac{1}{\cL_u}\int_0^u g_{\nu,D}(u-s)m(s)\dd s.
 \label{eq:linear-response}
\end{equation}
This formula is local in the usual response-theory sense. The slope $\cL_u$, diffusion scale $D_u$, resilience $\nu_u$, and boundary location are frozen over the response window used to compute the increment; their dependence on state and operational time is restored in the local closure below. Physically, the book slope converts imbalance into displacement, while the damped Abel kernel describes how past signed forcing is transmitted to the present boundary. This is the small-displacement, frozen-boundary reduction of the nonlinear Volterra front response derived in \cite{AngstmannGebbieJSTAT2026}, consistent with the latent-liquidity square-root impact mechanism \cite{Mastromatteo2014,Donier2015,Benzaquen2018}; here it is used to compute a second cumulant rather than a prescribed mean-impact path. Dependence on $(S,u)$ is restored after the local cumulant is computed.

\section{Green-function cumulant}
\label{sec:green-function-cumulant}
The Abel kernel is singular at zero, so the operational variance used in the closure is regularised at a mesoscopic scale $\Delta$, with $Y$ defined in Eq.~\eqref{eq:linear-response}:
\begin{equation}
 a_u^{(\Delta)}(S,u)=\frac{1}{\Delta}
 \Var\left[Y(u+\Delta)-Y(u)\mid S(u)=S\right].
 \label{eq:coarse-rate}
\end{equation}
This is the finite-scale transport coefficient carried through the remainder of the calculation. It is local in $(S,u)$, but it is not an unregularised microscopic limit; the scale $\Delta$ is part of the mesoscopic description.
By an unregularised microscopic limit we mean the formal instantaneous variance-rate limit obtained by sending the measurement scale to zero without retaining the finite operational window $\Delta$ or the short-scale regularisation; that limit is not the object used here \cite{AngstmannGebbieBSMFromLattice2026}.
Use the regularised kernel
\begin{equation}
 g_{\nu,D}^{(\tau_0)}(r)=
 \frac{\exp[-\nu(r+\tau_0)]}{\sqrt{4\pi D(r+\tau_0)}}\mathbf{1}_{r\ge0}
 \label{eq:regularized-kernel}
\end{equation}
and a locally stationary long-memory forcing covariance
\begin{equation}
 C_m(\tau;S,u)=A_m(S,u)(|\tau|+\tau_0)^{-\gamma(S,u)},
 \label{eq:forcing-cov}
\end{equation}
where $ 0<\gamma<1$.
The sign-memory exponent may be interpreted through the hidden-order splitting mechanism of Lillo, Mike and Farmer \cite{LilloMikeFarmer2005}. Identifying that event-time exponent with the operational-time signed-forcing exponent assumes a regular event-to-operational mesh and suitable forcing-amplitude statistics; it does not identify the memory exponent of calendar-time observables after subordination \cite{AngstmannGebbieJSTAT2026}. 

With the Fourier convention 
\begin{equation}
    C_m(\tau)=(2\pi)^{-1}\int e^{\ii\omega\tau}S_m(\omega)\dd\omega,
    \label{eq:fourier-convention}
\end{equation}
the exact coarse-grained spectral form is
\begin{equation}
 a_u^{(\Delta)}=\frac{1}{\Delta\cL_u^2}\frac{1}{2\pi}
 \int_{-\infty}^{\infty}4\sin^2\left(\frac{\omega\Delta}{2}\right)
 \left|\wh g_{\nu,D}^{(\tau_0)}(\omega)\right|^2S_m(\omega)\dd\omega.
 \label{eq:spectral-cumulant}
\end{equation}
This is the standard Wiener--Khinchin/filtering representation for the variance of a linear filter of a locally stationary forcing process. Here the increment
\(Y(u+\Delta)-Y(u)\) from Eq.~\eqref{eq:coarse-rate} contributes the transfer function \(e^{i\omega\Delta}-1\), whose squared modulus is \(4\sin^2(\omega\Delta/2)\); the boundary response contributes \(|\widehat g(\omega)|^2\); and the forcing contributes the spectrum \(S_m(\omega)\) \cite{Priestley1981}. The formula is exact for the frozen local coefficients and the regularised kernel. The subsequent approximation enters only when the finite-cutoff spectrum and finite-cutoff response are replaced by their low-frequency, zero-cutoff asymptotic forms.

\section{Asymptotic closure}
\label{sec:asymptotic-closure}

For $\Delta/\tau_0\gg1$, the low-frequency forcing spectrum is
\begin{equation}
 \frac{S_m(\omega)}{A_m}\sim C_\gamma |\omega|^{\gamma-1},
 \qquad
 C_\gamma=2\Gamma(1-\gamma)\sin\left(\frac{\pi\gamma}{2}\right).
 \label{eq:spectrum-asymp}
\end{equation}
Here low frequency refers to the small-$|\omega|$ behaviour of the spectral density $S_m(\omega)$, so Eq.~\eqref{eq:spectrum-asymp} is understood asymptotically as $|\omega|\to0$ \cite{Beran1994,Priestley1981}.
The microstructural motivation for the exponent is order splitting \cite{LilloMikeFarmer2005}, while the transform itself is the usual long-memory spectral calculation \cite{Beran1994}.
The zero-cutoff kernel from Eq.~\eqref{eq:spectral-cumulant} has
\begin{equation}
 \wh g_{\nu,D}(\omega)=\frac{1}{2\sqrt{D}}(\nu+
 i\omega)^{-1/2},
 \end{equation}
such that the power spectrum is
\begin{equation}
 |\wh g_{\nu,D}|^2=\frac{1}{4D\sqrt{\nu^2+\omega^2}}.
 \label{eq:g-transform}
\end{equation}
The transform follows from the Gamma-integral/Laplace transform of the
Abel kernel $t^{-1/2}\mathbf 1_{t>0}$ (or equivalently the Riemann--Liouville fractional-integral kernel), placing the response calculation in the same analytical family as fractional and long-memory filtering methods \cite{MillerRoss1993,Beran1994,Priestley1981}.
This gives the asymptotic closure
\begin{equation}
 a_u^{(\Delta)}(S,u)
 \simeq
 \left[\frac{A_m(S,u)}{\cL_u(S,u)^2D_u(S,u)}\right]
 \Delta^{-\gamma(S,u)}
 \Fcal_{\gamma(S,u)}(\nu_u(S,u)\Delta).
 \label{eq:closure-before-Aeff}
\end{equation}
The powers and constants in Eq.~\eqref{eq:closure-before-Aeff} follow directly from substituting Eqs.~\eqref{eq:spectrum-asymp} and \eqref{eq:g-transform} into Eq.~\eqref{eq:spectral-cumulant} and setting $\xi=\omega\Delta$. The factor $|\omega|^{\gamma-1}$ contributes $\Delta^{1-\gamma}$, the response denominator contributes
\[
 \frac{\Delta}{\sqrt{(\nu\Delta)^2+\xi^2}},
\]
and the outside rate normalisation contributes $1/\Delta$. The net scale is therefore $\Delta^{-\gamma}$, with the remaining dimensionless integral collected into $\Fcal_\gamma(\nu\Delta)$.
Defining the effective signed-forcing intensity
\begin{equation}
 \Aeff(S,u)=\frac{A_m(S,u)}{D_u(S,u)},
 \label{eq:Aeff}
\end{equation}
one obtains the local mesoscopic operational variance rate
\begin{equation}
 a_u^{(\Delta)}(S,u)
 \simeq
 \left[\frac{\Aeff(S,u)}{\cL_u(S,u)^2}\right]
 \Delta^{-\gamma(S,u)}
 \Fcal_{\gamma(S,u)}(\nu_u(S,u)\Delta).
 \label{eq:operational-result}
\end{equation}
The bracketed factor makes the inverse-liquidity dependence explicit: stronger effective signed forcing increases the boundary variance, while a steeper local book suppresses boundary motion quadratically. The factor $\Delta^{-\gamma}$ carries the coarse-graining dependence induced by long-memory signed order flow, and the dimensionless argument $\nu_u\Delta$ compares the resilience time scale of the book with the operational scale over which the boundary increment is measured.

With the normalization used in Eq.~\eqref{eq:operational-result}, the response function may be written as the dimensionless spectral integral
\begin{equation}
 \Fcal_\gamma(z)=
 \frac{C_\gamma}{2\pi}
 \int_{-\infty}^{\infty}
 \frac{\sin^2(\xi/2)|\xi|^{\gamma-1}}
 {\sqrt{z^2+\xi^2}}\,\dd\xi,
 \qquad 0<\gamma<1,
 \label{eq:F-definition}
\end{equation}
up to the same zero-cutoff asymptotic normalization used in passing from Eq.~\eqref{eq:spectral-cumulant} to Eq.~\eqref{eq:closure-before-Aeff}. The response function is not fitted; it is constrained by asymptotics:
\begin{equation}
 \Fcal_\gamma(z)\sim
 \begin{cases}
 \displaystyle \frac{\tan(\pi\gamma/2)}{2(1-\gamma)},& z\ll1,\\[1.2em]
 \displaystyle K_\gamma^{(0)}z^{\gamma-1},& z\gg1,
 \end{cases}
 \label{eq:F-asymp}
\end{equation}
where
\begin{equation}
 K_\gamma^{(0)}=\frac{\Gamma(1-\gamma)\sin(\pi\gamma/2)\Gamma(\gamma/2)\Gamma((1-\gamma)/2)}{2\pi^{3/2}}.
 \label{eq:Kgamma}
\end{equation}
The two limits in Eq.~\eqref{eq:F-asymp} follow from beta/gamma integral identities and asymptotic scaling of oscillatory integrals: for \(z\ll1\) the denominator is asymptotically \(|\xi|\), while for
\(z\gg1\) the scaling \(\xi=zy\) and oscillatory averaging of
\(\sin^2(zy/2)\) gives the displayed coefficient
\cite{GradshteynRyzhik,Wong2001}. The high-resilience prefactor is the zero-cutoff long-memory value; finite cutoff corrections can modify the prefactor but not the scaling.

\section{Clock projection and pricing}
\label{sec:clock-projection-pricing}

Let
\begin{equation}
 \Xiop(S,u;\Delta)=
 \left[\frac{\Aeff(S,u)}{\cL_u(S,u)^2}\right]
 \Delta^{-\gamma(S,u)}
 \Fcal_{\gamma(S,u)}(\nu_u(S,u)\Delta).
 \label{eq:Xi}
\end{equation}
This is the operational-time variance kernel obtained from the reaction-boundary cumulant. It is still an order-book object: it depends on signed-forcing intensity, local liquidity, resilience, memory, and the operational scale $\Delta$. Its use as a calendar-time local variance requires an admissible diffusion closure and a choice of clock and measure.

For an admissible diffusion closure and a deterministic activity clock $u=U(t)$, the benchmark projection is
\begin{equation}
\sigloc^2(S,t)=\alpha(t)\Xiop(S,U(t);\Delta).
 \label{eq:deterministic-localvol}
\end{equation}
Equation~\eqref{eq:deterministic-localvol} says that the clock does not change the Green-function cumulant itself. It only determines the rate at which operational boundary variance is sampled in calendar time. Higher activity therefore amplifies calendar-time variance, even when the operational response kernel is unchanged.

Under a risk-neutral measure, in the usual no-arbitrage pricing framework, the corresponding pricing PDE uses the risk-neutral version of the same projected kernel \cite{HarrisonPliska1981,DelbaenSchachermayer1994,Duffie2001}:
\begin{equation}
 V_t+(r-q)xV_x+\frac12\alpha(t)\Xiop^\Q(\log x,U(t);\Delta)x^2V_{xx}-rV=0.
 \label{eq:pde}
\end{equation}
Here $x$ is the traded price, so that the boundary state entering the operational kernel is $S=\log x$. The PDE is therefore a benchmark projection of the mesoscopic reaction-boundary variance into the usual local-volatility pricing form, not an independent postulate of a diffusion coefficient.

In log-price variables, the corresponding projected backward operator has the schematic risk-neutral form
\begin{equation}
 \Bop_t^U f(S)=b^\Q(S,t)\partial_S f(S)+\frac12\alpha(t)\Xiop^\Q(S,U(t);\Delta)\partial_{SS}f(S),
 \label{eq:backward-operator}
\end{equation}
with the usual killing term $-r f$ added for discounted claims. In the pure diffusion benchmark corresponding to Eq.~\eqref{eq:pde},
\[
 b^\Q(S,t)=r-q-\frac12\alpha(t)\Xiop^\Q(S,U(t);\Delta),
\]
subject to the usual convention changes if discounting or numeraires are absorbed into the state-price density. The forward density operator is the formal adjoint
\begin{equation}
 (\Bop_t^U)^\dagger p=-\partial_S\left(b^\Q(S,t)p\right)+\frac12\partial_{SS}\left(\alpha(t)\Xiop^\Q(S,U(t);\Delta)p\right),
 \label{eq:forward-adjoint}
\end{equation}
up to boundary and numeraire conventions. For a deterministic clock the multiplier $\alpha(t)$ is common to both sides, so the forward--backward adjoint structure is inherited from the operational generator. This inheritance is no longer automatic when the projection is stochastic, filtered, or non-unique.

The homogeneous benchmark freezes $\Aeff$, $\cL_u$, $\nu_u$, and $\gamma$:
\begin{equation}
 \sigma_{\rm eff}^2=\alpha_0
 \left[\frac{{\Aeff}_0}{\cL_0^2}\right]
 \Delta^{-\gamma_0}
 \Fcal_{\gamma_0}(\nu_0\Delta).
 \label{eq:homogeneous-result}
\end{equation}
This expression gives the homogeneous closure of the theory. Activity multiplies the operational variance kernel; effective signed forcing supplies boundary fluctuations; the local liquidity slope suppresses those fluctuations quadratically; resilience damps the contribution of persistent forcing; and the memory exponent controls the coarse-graining dependence through $\Delta^{-\gamma_0}$.

The formula is best stated measure-generically. Under an admissible measure $\M$,
\begin{equation}
 \Xiop_0^{\M}(\Delta)=
 \left[\frac{{\Aeff}_0^{\M}}{(\cL_0^{\M})^2}\right]
 \Delta^{-\gamma_0^{\M}}
 \Fcal_{\gamma_0^{\M}}(\nu_0^{\M}\Delta).
 \label{eq:measure-generic}
\end{equation}
The pricing PDE uses the $\Q$ version. This avoids identifying physical order-flow statistics with risk-neutral quantities before specifying a risk-premium model. This is the point at which model uncertainty enters in a way comparable to robust or uncertain-model valuation: the same structural kernel may support different pricing systems once the admissible measure and projection class are varied \cite{DenisMartini2006,Cont2006}. In general,
\begin{equation}
 \Xiop^{\mathbb P}(S,u;\Delta)\ne \Xiop^{\mathbb Q}(S,u;\Delta),
 \label{eq:P-not-Q-kernel}
\end{equation}
because the signed-forcing intensity, liquidity response, resilience, memory exponent, or clock-risk compensation may change under the pricing measure. The derivation supplies the operational functional form; a pricing application must specify, estimate, or constrain the transformation from the physical kernel to the risk-neutral kernel.

\section{Non-unique time and adjoint-real projections}
\label{sec:adjoint-real}
The non-unique-time extension keeps the operational kernel \eqref{eq:Xi} and changes only the projection. For a family of clocks $\mathfrak U$,
\begin{equation}
 \Sigma_{\mathfrak U}^2(S,t)=\mathcal P_{\mathfrak U}[\Xiop(S,u;\Delta);t].
 \label{eq:NUT-projection}
\end{equation}
The operator $\mathcal P_{\mathfrak U}$ denotes a projection from operational to calendar time. For a deterministic absolutely continuous clock it reduces to the multiplier in Eq.~\eqref{eq:deterministic-localvol}, namely $\mathcal P_{\mathfrak U}[\Xiop;t]=\alpha(t)\Xiop(S,U(t);\Delta)$. For a stochastic, filtered, or non-unique clock it represents the corresponding conditional projection over admissible operational-time increments under the chosen pricing measure.

\begin{definition}[Projected pricing pair]
Given an operational kernel $\Xiop(S,u;\Delta)$ and a clock projection $U\in\mathfrak U$, a projected pricing pair is a pair $(\Bop_t^U,\Gop_t^U)$, where $\Bop_t^U$ acts backward on claims and $\Gop_t^U$ acts forward on state laws. The pair is interpreted as a one-state pricing representation only when both operators are defined on compatible domains and are associated with the same projected transition mechanism.
\end{definition}

\begin{definition}[Adjoint-real clock]
A clock projection $U\in\mathfrak U$ is called adjoint-real if its projected pair satisfies
\begin{equation}
 \Pair{p}{\Bop_t^U V}=\Pair{\Gop_t^U p}{V}
 \label{eq:adjoint-real-condition}
\end{equation}
for all admissible densities $p$ and claims $V$, after including the appropriate discounting, boundary, killing, and numeraire terms. The set of adjoint-real clocks is
\begin{equation}
 \Ad(\mathfrak U)=\{U\in\mathfrak U:\; \Gop_t^U=(\Bop_t^U)^\dagger\;\text{on the chosen state space}\}.
 \label{eq:Ad-set}
\end{equation}
\end{definition}

\begin{proposition}[Adjoint consistency as a reality constraint] \label{prop:adjoint-real}
A non-unique operational-to-calendar projection defines a coherent one-state pricing system only if its clock belongs to $\Ad(\mathfrak U)$. A projection that maps the operational variance kernel into a plausible backward pricing coefficient, but does not induce the corresponding adjoint forward law on the same state space and information set, is a formal coefficient projection rather than a pricing model.
\end{proposition}

\begin{proof}[Proof outline] \label{proof:adjoint-realoutline}
A one-state risk-neutral pricing representation requires a single projected transition mechanism. Acting backward on claims this mechanism gives a valuation semigroup and its infinitesimal operator $\Bop_t^U$; acting forward on state laws it gives the density evolution and its operator $\Gop_t^U$. Consistency of valuation and transition probabilities requires the bilinear pairing between densities and claims to be preserved, up to the usual discounting and cash-flow terms. Differentiating this pairing gives Eq.~\eqref{eq:adjoint-real-condition}. If the projected backward coefficient and the projected forward law are not adjoints, the representation prices claims under one calendar-time mechanism while evolving probabilities under another. Such a construction does not define a coherent one-state pricing model.
\end{proof}

The expanded form of the proof outlined above is provided as a consistency theorem in \ref{app:adjoint-consistency}.

\begin{remark}[A hidden-clock counterexample]
Suppose the clock has an unobserved state $Z_t$. Over a short calendar
interval $[t,t+\Delta t]$, a one-state backward projection may use the coefficient
\begin{equation}
 \Sigma_{\Delta t}^2(S,t)
 =
 \frac{1}{\Delta t}
 \E^\Q\!\left[
   \int_{U(t)}^{U(t+\Delta t)}
   \Xiop(S,u;\Delta)\dd u
   \,\middle|\, S_t=S
 \right].
\end{equation}
If $U$ is absolutely continuous this reduces formally to
$\E^\Q[\dot U_t\Xiop(S,U_t;\Delta)\mid S_t=S]$ as $\Delta t \downarrow0$.
Thus the displayed coefficient can give a plausible scalar pricing PDE.
If the law of future clock increments depends on $Z_t$ beyond $S_t$, the original hidden-state mechanism need not be closed on $S_t$. A scalar coefficient may nevertheless reproduce selected marginal quantities without reproducing the full hidden-state pricing semigroup. A coherent one-state representation therefore requires an independently valid common projected mechanism. If no such mechanism exists within the selected projection class, the projection must be rejected or the state description enlarged.
\end{remark}

A maturity-level object can now be written as
\begin{equation}
 \V_{t,T}^{\mathfrak U}=\E^\Q\left[\int_{U(t)}^{U(T)}\Xiop^\Q(S(u),u;\Delta)\dd u\mid\mathcal F_t\right],
 \label{eq:NUT-IV}
\end{equation}
provided that the expectation is taken under a pricing measure and clock projection for which the induced forward--backward pair is adjoint-real. When $U(t)$ is deterministic, \eqref{eq:NUT-IV} reduces, for a state-dependent kernel, to
\[
 \E^\Q\left[\int_t^T \alpha(s)\Xiop^\Q(S_s,U(s);\Delta)\dd s\mid\mathcal F_t\right],
\]
where $S_s$ is the projected state process. The simpler pointwise benchmark $\int_t^T\alpha(s)\Xiop^\Q(S,U(s);\Delta)\dd s$ is recovered when the state is frozen, the kernel is state independent, or the expression is used as a deterministic local-variance evaluation at a fixed state. If several elements of $\Ad(\mathfrak U)$ survive and induce non-equivalent valuations, non-unique time remains as incompleteness: prices can depend on the chosen adjoint-real projection or on the chosen pricing measure over clock risk. If $\Ad(\mathfrak U)$ is empty for the chosen state space, the state space must be enlarged or the projection class rejected \cite{AngstmannGebbie2026NonUniqueTime}.

\section{Local volatility, time changes, and projection}

The deterministic-clock PDE in Eq.~\eqref{eq:pde} has the same algebraic form as a local-volatility pricing equation. The interpretation is different. In a standard local-volatility model the coefficient is specified directly in calendar time or inferred from option prices. Here it is the projection of an operational reaction-boundary cumulant, and the operational kernel can be discussed before any calendar-time clock is selected.

The asymptotic form of $\Xiop$ is what makes the adjoint condition concrete. Non-unique time does not change the Green-function cumulant itself: it changes how the scale-dependent kernel, with its explicit $\Delta^{-\gamma}\Fcal_\gamma(\nu\Delta)$ dependence, is sampled and conditioned in calendar time. Thus two clocks can agree on a scalar projected variance at a point while implying different forward laws once hidden activity, filtration, or clock-risk information is taken into account. The local-volatility equation should therefore be read as a candidate projection of the asymptotic kernel, not as a primitive replacement for the operational model.

The construction is also distinct from a standard time-changed process. A standard time-change model specifies a deterministic or stochastic clock and then studies the induced calendar-time process \cite{Bochner1949,Clark1973,CarrWu2004}. Here the clock is part of the modelling ambiguity. The family $\mathfrak U$ contains candidate projections, whereas $\Ad(\mathfrak U)$ contains only those projections that preserve the forward--backward adjoint relation needed for pricing. This places the clock choice closer to a model-uncertainty problem than to a purely parametric calibration problem \cite{DenisMartini2006,Cont2006}.

Finally, the adjoint-real condition is not the same as a Markovian projection or marginal mimicking result. Markovian projection asks whether a possibly non-Markovian calendar-time process can be represented by a Markov process with matching marginal distributions \cite{Gyongy1986,BrunickShreve2013}. The present question is earlier: given an operational kernel, which projections into calendar time produce mutually consistent forward and backward operators? The answer is not supplied by the variance projection alone. It is supplied by adjoint consistency. Thus the asymptotic form of the boundary kernel makes the candidate local variance explicit, while the adjoint condition provides a necessary consistency condition for this candidate to live on the proposed state space as a pricing model.

\section{Reaction-boundary volatility surfaces}
\label{sec:volatility-simulations}

The analytic closure above can also be used to illustrate how the structural reaction-boundary quantities appear after projection into a local-volatility-like surface. We consider stylised liquidity, resilience-response and signed-forcing profiles across log-moneyness and maturity, using the finite-scale operational variance derived in Sections~\ref{sec:green-function-cumulant}--\ref{sec:asymptotic-closure} and its deterministic calendar-time projection from Section~\ref{sec:clock-projection-pricing}. The six simulations below are organised as comparisons between liquidity-only, resilience-response-only and combined structural profiles in the low- and high-$z$ regimes. Each case is shown as a matched three-dimensional surface and contour pair so that the overall morphology and the corresponding level-set structure can be read together. The combined high-$z$ case remains the most complete numerical example. 

At each grid point we evaluate the operational kernel in Eq.~\eqref{eq:Xi}, using the leading low- or high-$z$ asymptotic in Eq.~\eqref{eq:F-asymp}. The deterministic activity-clock projection is then the one already given in Eq.~\eqref{eq:deterministic-localvol}. The simulations preserve the sequence used throughout the paper: the order-book quantities determine the operational variance first, and the activity clock determines how that variance is represented in calendar time.

For plotting we use the log-moneyness coordinate
\begin{equation}
    m=\ln(x/x_0),
    \qquad
    S=\ln x=\ln x_0+m,
    \label{eq:sim-log-moneyness}
\end{equation}
where $x_0$ is the reference price. The displayed quantity is
\begin{equation}
    \sigloc(m,T)
    =
    \left[
      \alpha(T)\,
      \Xiop\!\left(\ln x_0+m,U(T);\Delta\right)
    \right]^{1/2},
    \label{eq:sim-plotted-volatility}
\end{equation}
evaluated on a rectangular grid in log-moneyness $m$ and maturity $T$. The coordinate $m$ is used here to make the state dependence comparable with the familiar representation of option-volatility surfaces.

The two numerical regimes are the asymptotic limits already given in Eq.~\eqref{eq:F-asymp}. In the low-$z$ regime the leading response is independent of $z$ when $\gamma$ is fixed. The low-$z$ resilience-response-only calculation is therefore deliberately diagnostic: variation in $z(m)$ alone cannot generate leading-order state dependence. In the high-$z$ regime the response instead scales as $K_\gamma^{(0)}z^{\gamma-1}$. Since $0<\gamma<1$, a smaller value of $z=\nu\Delta$ raises this response multiplier. The comparison between the two regimes therefore separates effects that originate in the liquidity and forcing profiles from those generated directly by resilience.

For these illustrations we hold $\gamma$ constant and introduce a simple maturity-dependent deterministic activity profile,
\begin{equation}
    \alpha(T)
    =
    \alpha_0\left(
      1+a_{\mathrm{short}}e^{-T/\tau_\alpha}
    \right).
    \label{eq:sim-activity-profile}
\end{equation}
The state dependence is then introduced through the local liquidity slope $\cL(m)$, the dimensionless resilience-response argument $z(m)=\nu(m)\Delta$, and the effective signed-forcing intensity $A_{\mathrm{eff}}(m)$:
\begin{align}
    \cL(m)
    &=
    L_0\max\!\left(
      L_{\mathrm{floor}},
      1+L_{\mathrm{skew}}m+L_{\mathrm{smile}}m^2
    \right),
    \label{eq:sim-liquidity-profile}\\
    z(m)
    &=
    z_{\mathrm{regime}}
    e^{\nu_{\mathrm{skew}}m},
    \label{eq:sim-z-profile}\\
    \Aeff(m)
    &=
    A_{\mathrm{eff},0}
    e^{-A_{\mathrm{stress}}m}.
    \label{eq:sim-forcing-profile}
\end{align}
Positive $L_{\mathrm{skew}}$ reduces liquidity on the downside, positive $\nu_{\mathrm{skew}}$ lowers the resilience-response argument $z(m)$ on the downside, and positive $A_{\mathrm{stress}}$ increases the effective signed-forcing scale there. These profiles are phenomenological: their purpose is to switch the structural regimes on and off in a controlled way rather than to represent estimated market quantities.

Table~\ref{tab:sim-parameter-ranges} collects the dimensionless baseline values and ranges used in the numerical comparison. They set the scale for comparison between profiles and response regimes and should not be read as empirical estimates.

\begin{table}[!tbp]
\centering
\scriptsize
\caption{Dimensionless baseline values and simulation ranges.}
\label{tab:sim-parameter-ranges}
\begin{tabular}{p{0.08\linewidth}p{0.05\linewidth}p{0.05\linewidth}p{0.05\linewidth}p{0.55\linewidth}}
\toprule
Para. & Base & Min. & Max. & Role \\
\midrule
\(\gamma\) & 0.55 & 0.30 & 0.80 & Long-memory; avoids endpoints 0 and 1. \\
\(\Delta\) & 1.0 & 0.5 & 2.0 & Operational coarse-graining scale. \\
\(A_{\mathrm{eff},0}\) & 1.0 & 0.5 & 2.0 & Effective signed-forcing intensity scale. \\
\(L_0\) & 1.0 & 0.75 & 1.5 & Baseline liquidity slope. \\
\(z_{\mathrm{low}}\) & 0.05 & 0.01 & 0.10 & Low-\(z\) response regime. \\
\(z_{\mathrm{high},0}\) & 10.0 & 5.0 & 25.0 & High-\(z\) response regime. \\
\(\alpha_0\) & 1.0 & 0.5 & 2.0 & Baseline deterministic activity. \\
\(a_{\mathrm{short}}\) & 0.25 & 0.0 & 0.75 & Short-maturity activity lift. \\
\(\tau_\alpha\) & 0.50 & 0.25 & 1.0 & Activity decay scale in years. \\
\(L_{\mathrm{skew}}\) & 0.80 & 0.0 & 1.5 & Lower downside liquidity when positive. \\
\(L_{\mathrm{smile}}\) & 1.50 & 0.0 & 3.0 & Liquidity variation from the centre. \\
\(\nu_{\mathrm{skew}}\) & 1.00 & 0.0 & 2.0 & Weaker downside resilience when \(>0\). \\
\(A_{\mathrm{stress}}\) & 0.75 & 0.0 & 1.5 & Stronger downside forcing in the combined profile. \\
\(L_{\mathrm{floor}}\) & 0.20 & 0.10 & 0.50 & Positive floor for the liquidity profile. \\
\bottomrule
\end{tabular}
\end{table}

We first isolate the liquidity channel. Figures~\ref{fig:sim-liquidity-low-z} and~\ref{fig:sim-liquidity-high-z} use the same state-dependent liquidity profile while holding the other structural parameters fixed. Their comparison therefore shows how the inverse-square liquidity dependence enters the projected surface in the two response regimes.

\begin{figure}[!tbp]
\centering
\begin{minipage}[t]{0.48\linewidth}
    \centering
    \includegraphics[width=\linewidth]{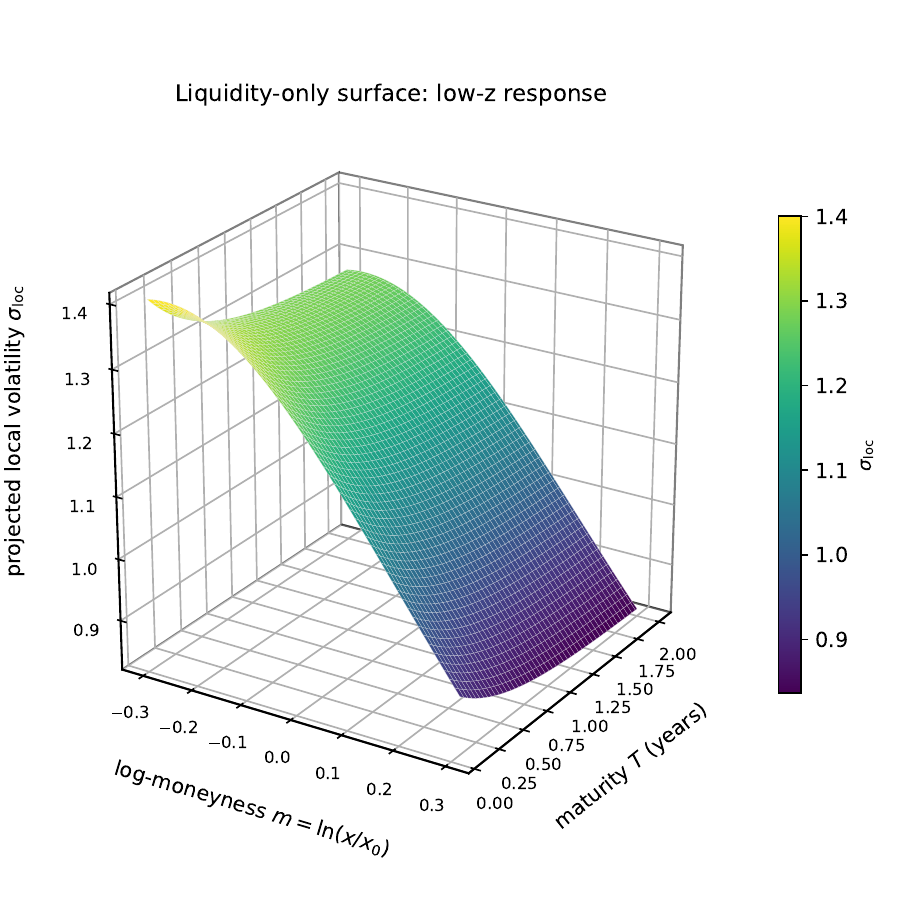}
    \textnormal{(a)}
\end{minipage}
\hfill
\begin{minipage}[t]{0.48\linewidth}
    \centering
    \includegraphics[width=\linewidth]{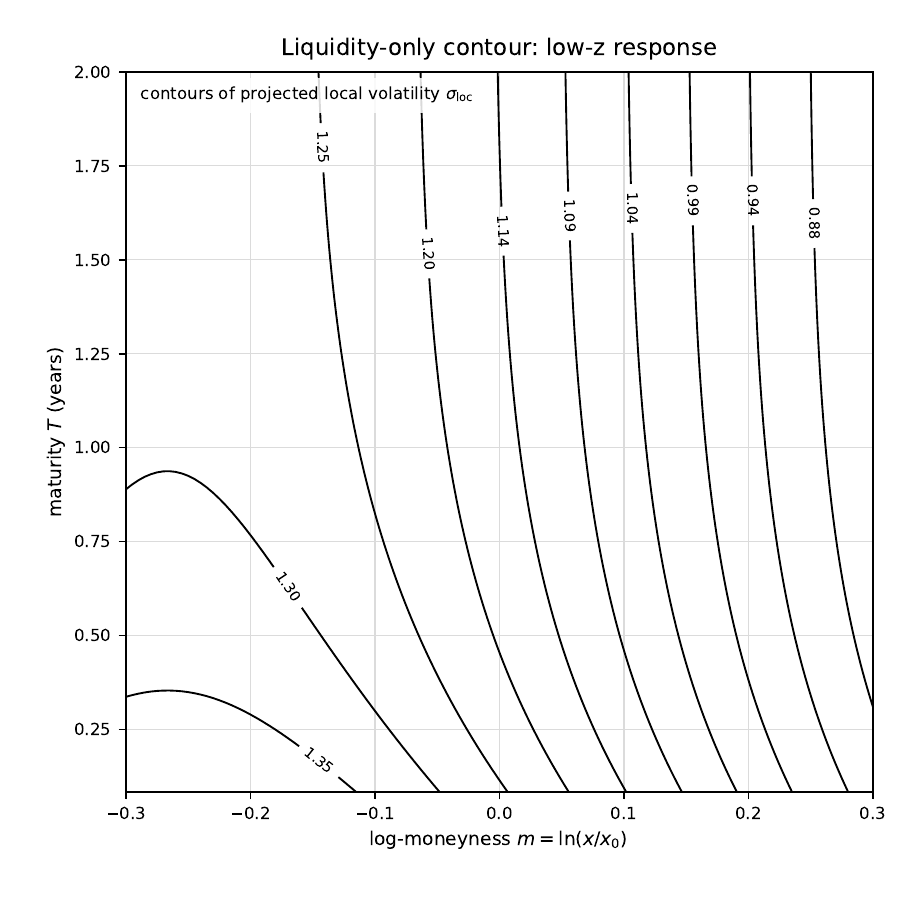}
    \textnormal{(b)}
\end{minipage}
\caption{Liquidity-only simulation in the low-$z$ response regime. (a) Projected structural volatility surface and (b) the corresponding contour representation. The only structural source of state dependence is the liquidity profile $\cL(m)$, which enters the closure through the inverse-square factor $1/\cL(m)^2$. Since the low-$z$ response approximation in Eq.~\eqref{eq:F-asymp} is independent of $z$ at fixed $\gamma$, the response channel does not contribute leading-order state dependence. This shows that the surface morphology in this regime is generated by the inverse-square liquidity together with the maturity-dependent activity profile.}
\label{fig:sim-liquidity-low-z}
\end{figure}

\begin{figure}[!tbp]
\centering
\begin{minipage}[t]{0.48\linewidth}
    \centering
    \includegraphics[width=\linewidth]{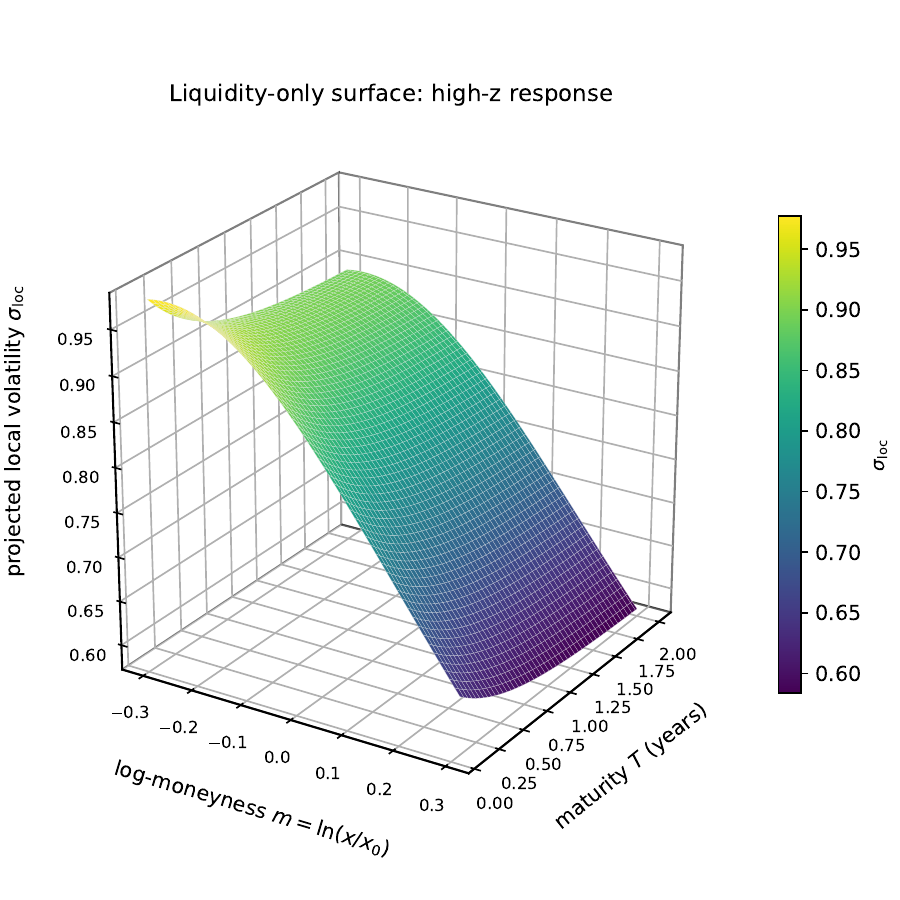}
    \textnormal{(a)}
\end{minipage}
\hfill
\begin{minipage}[t]{0.48\linewidth}
    \centering
    \includegraphics[width=\linewidth]{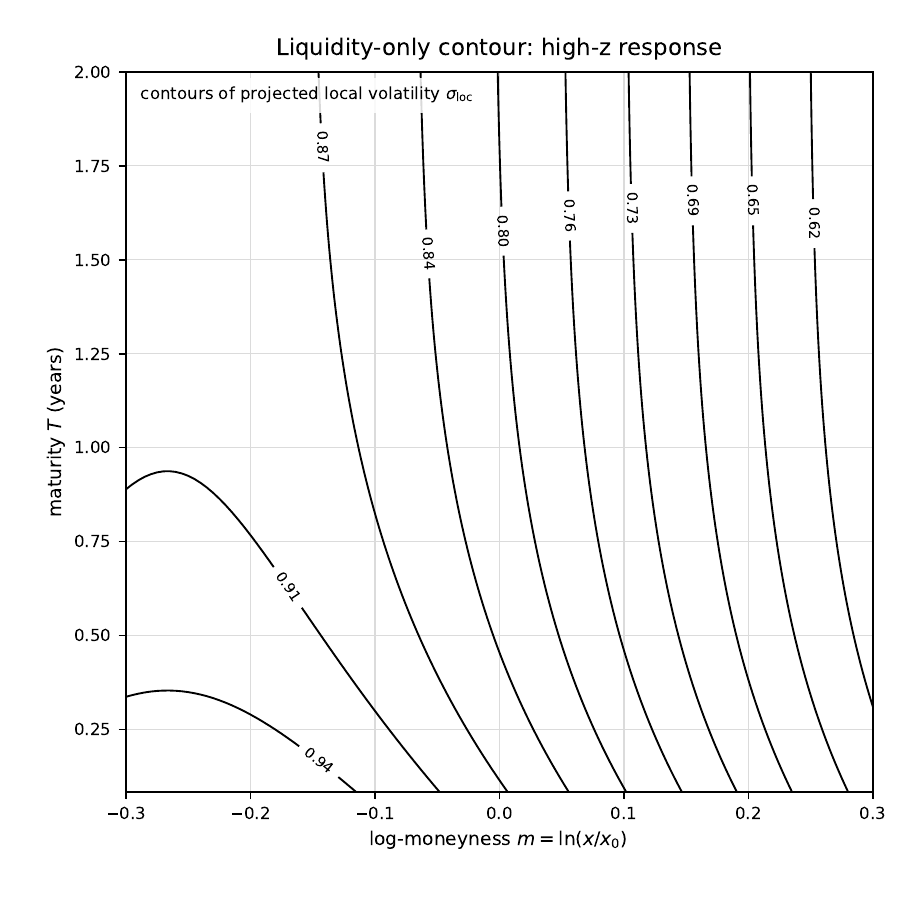}
    \textnormal{(b)}
\end{minipage}
\caption{Liquidity-only simulation in the high-$z$ response regime. (a) Projected structural volatility surface and (b) the corresponding contour representation. The liquidity profile remains the only source of state dependence because $z$ is held constant in this profile. Comparison with Fig.~\ref{fig:sim-liquidity-low-z} isolates the change in response-regime scale while preserving the same liquidity-driven surface shape. We see that changing the response regime changes the volatility scale without introducing a new source of state-dependent morphology.}
\label{fig:sim-liquidity-high-z}
\end{figure}

\FloatBarrier

We next hold liquidity and forcing fixed and vary only the resilience-response profile. Figures~\ref{fig:sim-resilience-low-z} and~\ref{fig:sim-resilience-high-z} make the asymptotic distinction in Eq.~\eqref{eq:F-asymp} visible. The low-$z$ case provides a diagnostic control, whereas the high-$z$ case allows the state dependence of $z(m)$ to enter directly through the response multiplier.

\begin{figure}[!tbp]
\centering
\begin{minipage}[t]{0.48\linewidth}
    \centering
    \includegraphics[width=\linewidth]{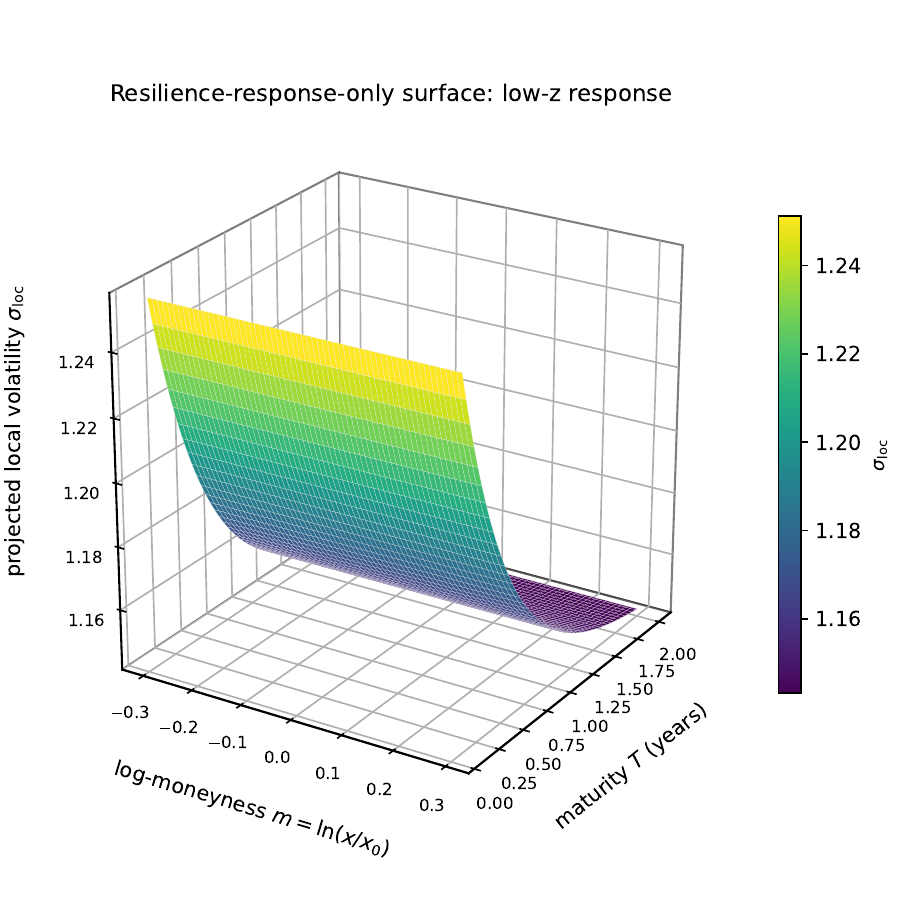}
    \textnormal{(a)}
\end{minipage}
\hfill
\begin{minipage}[t]{0.48\linewidth}
    \centering
    \includegraphics[width=\linewidth]{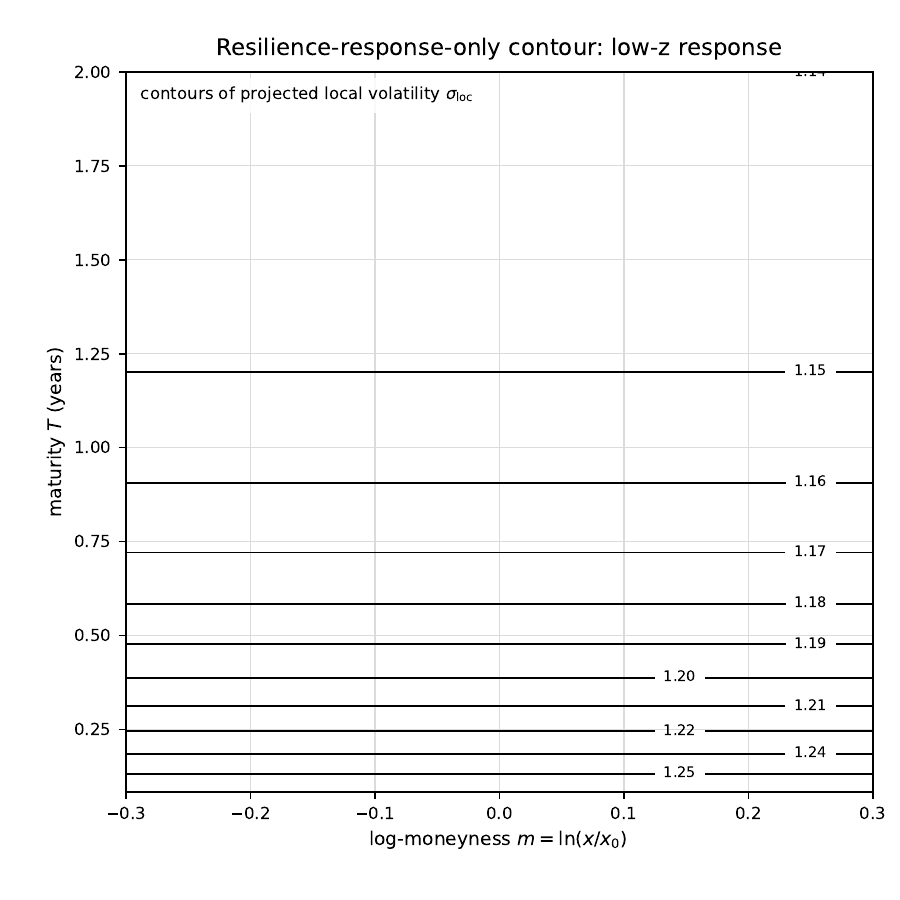}
    \textnormal{(b)}
\end{minipage}
\caption{Resilience-response-only simulation in the low-$z$ response regime. (a) Projected structural volatility surface and (b) the corresponding contour representation. The profile varies $z(m)=\nu(m)\Delta$, but the leading low-$z$ approximation in Eq.~\eqref{eq:F-asymp} does not depend on $z$ at fixed $\gamma$. The figure is therefore a diagnostic control: it shows the maturity activity effect without additional leading-order response-driven state dependence. We see explicitly that variation in resilience alone does not generate leading-order state dependence in the low-$z$ approximation.}
\label{fig:sim-resilience-low-z}
\end{figure}

\begin{figure}[!tbp]
\centering
\begin{minipage}[t]{0.48\linewidth}
    \centering
    \includegraphics[width=\linewidth]{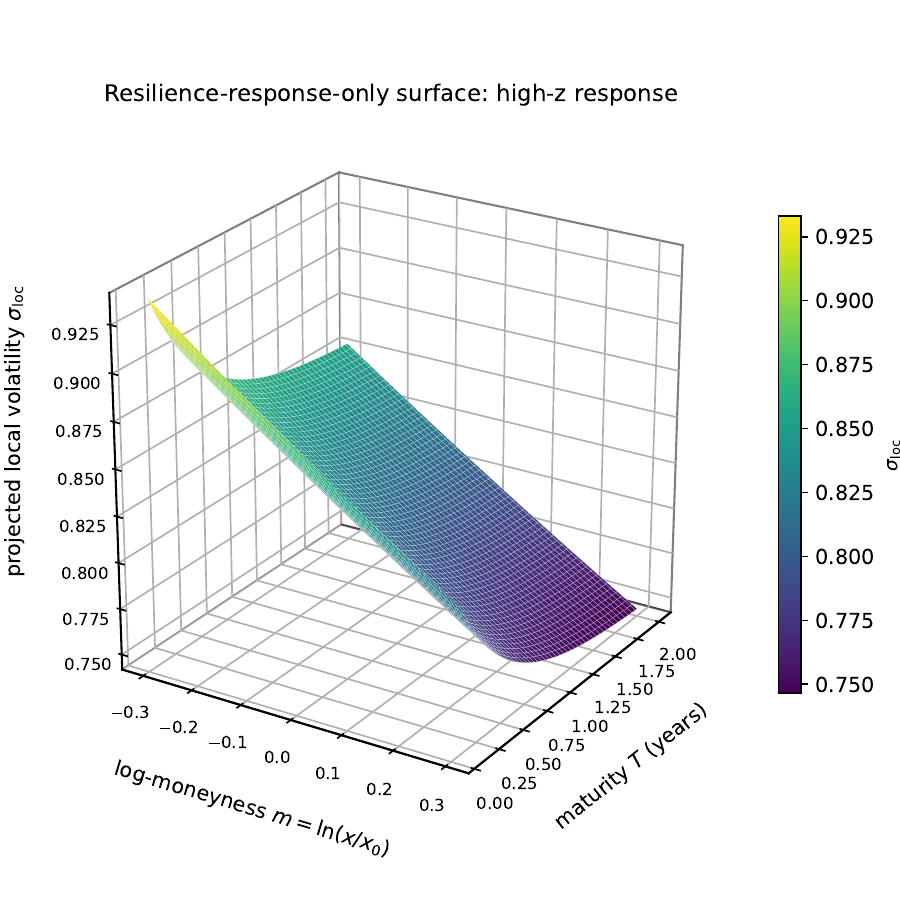}
    \textnormal{(a)}
\end{minipage}
\hfill
\begin{minipage}[t]{0.48\linewidth}
    \centering
    \includegraphics[width=\linewidth]{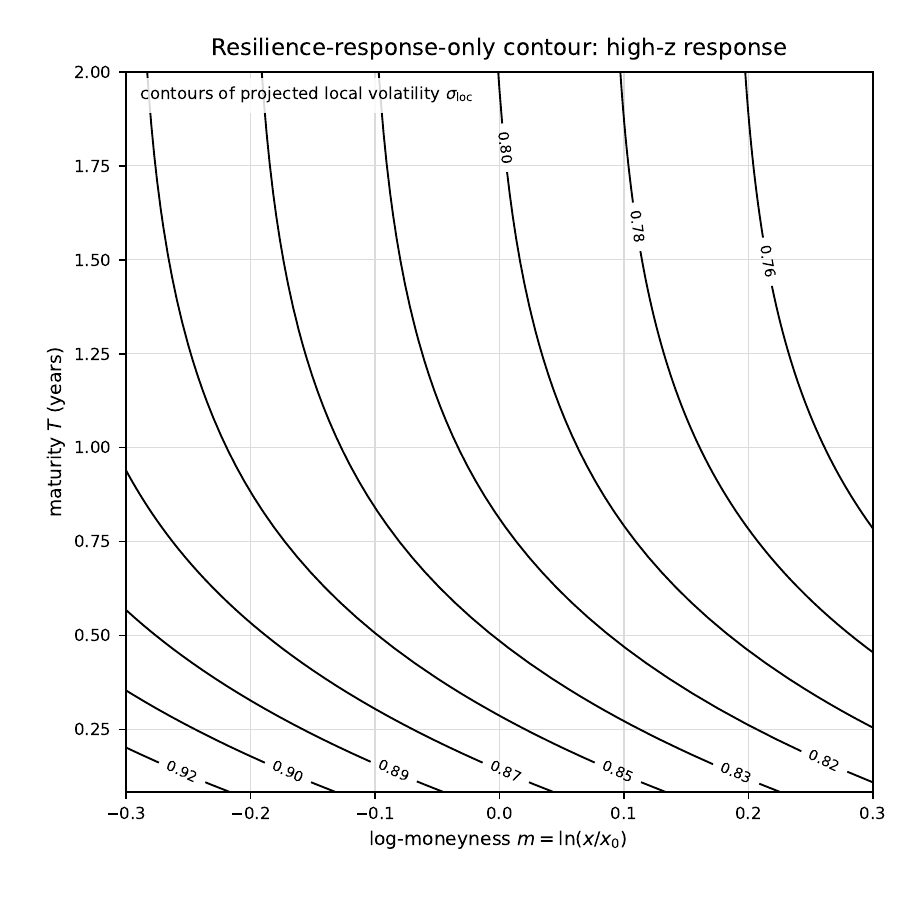}
    \textnormal{(b)}
\end{minipage}
\caption{Resilience-response-only simulation in the high-$z$ response regime. (a) Projected structural volatility surface and (b) the corresponding contour representation. Here $z(m)$ enters directly through $K_\gamma^{(0)}z^{\gamma-1}$. With $\gamma<1$, lower downside $z(m)$ increases the response multiplier, producing state dependence even though liquidity and forcing are held fixed. This shows how resilience-response variation becomes a direct source of surface morphology in the high-$z$ regime.}
\label{fig:sim-resilience-high-z}
\end{figure}

\FloatBarrier

The final comparison restores the liquidity and forcing profiles together with the resilience-response channel. Figures~\ref{fig:sim-combined-low-z} and~\ref{fig:sim-combined-high-z} show the two combined cases. In the low-$z$ branch the leading response remains independent of $z$, whereas the high-$z$ branch restores the direct resilience-response contribution.

\begin{figure}[!tbp]
\centering
\begin{minipage}[t]{0.48\linewidth}
    \centering
    \includegraphics[width=\linewidth]{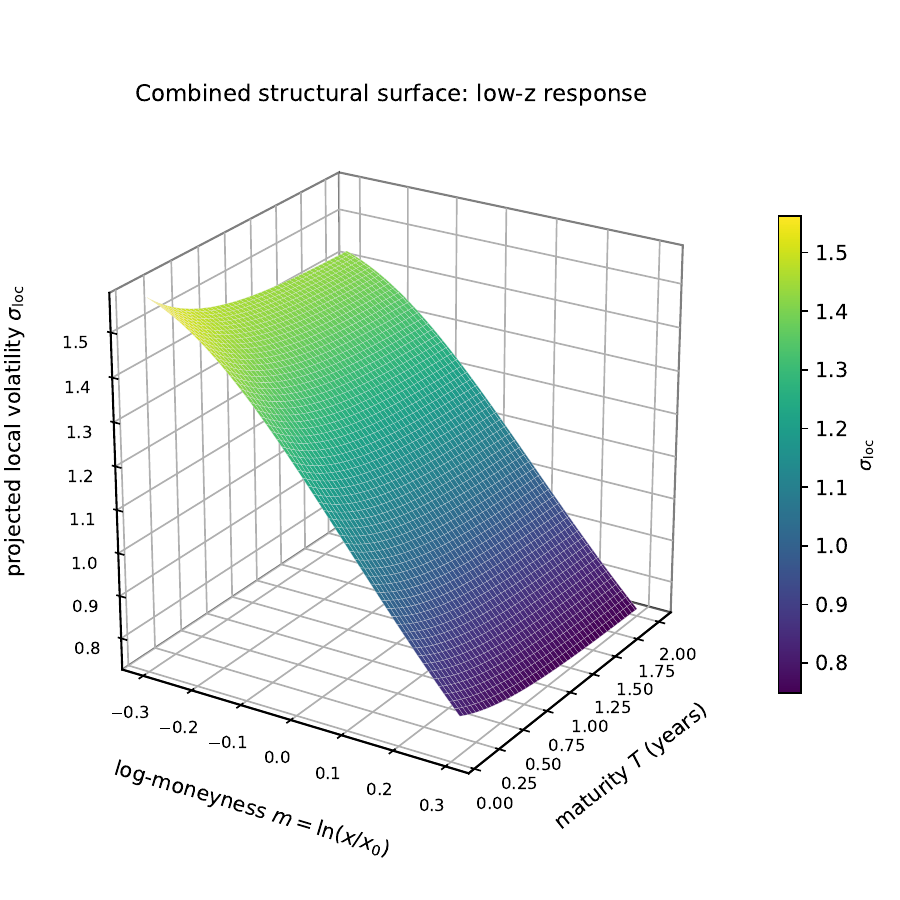}
    \textnormal{(a)}
\end{minipage}
\hfill
\begin{minipage}[t]{0.48\linewidth}
    \centering
    \includegraphics[width=\linewidth]{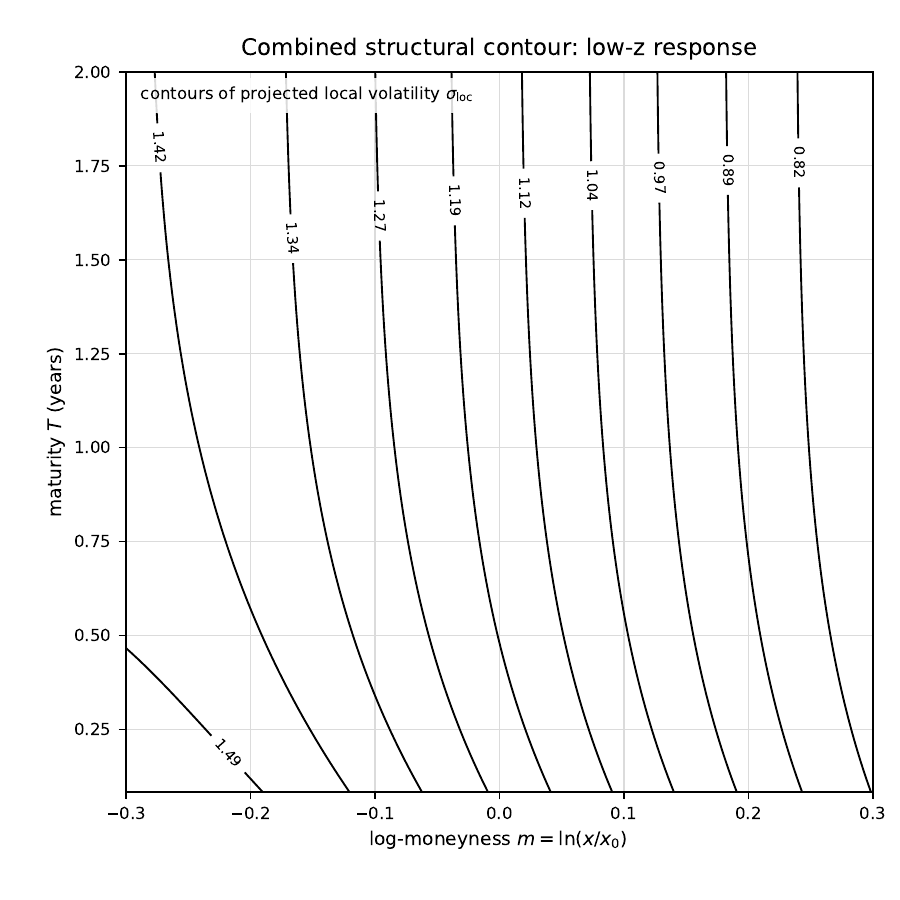}
    \textnormal{(b)}
\end{minipage}
\caption{Combined structural simulation in the low-$z$ response regime. (a) Projected structural volatility surface and (b) the corresponding contour representation. The surface combines liquidity variation and downside forcing while the leading response approximation in Eq.~\eqref{eq:F-asymp} remains independent of $z$ at fixed $\gamma$. The resulting shape should be interpreted as a stylised liquidity-and-forcing effect rather than as a calibrated volatility surface. We see that combining these two active channels is already sufficient to produce non-trivial state- and maturity-dependent surface morphology before resilience enters at leading order.}
\label{fig:sim-combined-low-z}
\end{figure}

\begin{figure}[!tbp]
\centering
\begin{minipage}[t]{0.48\linewidth}
    \centering
    \includegraphics[width=\linewidth]{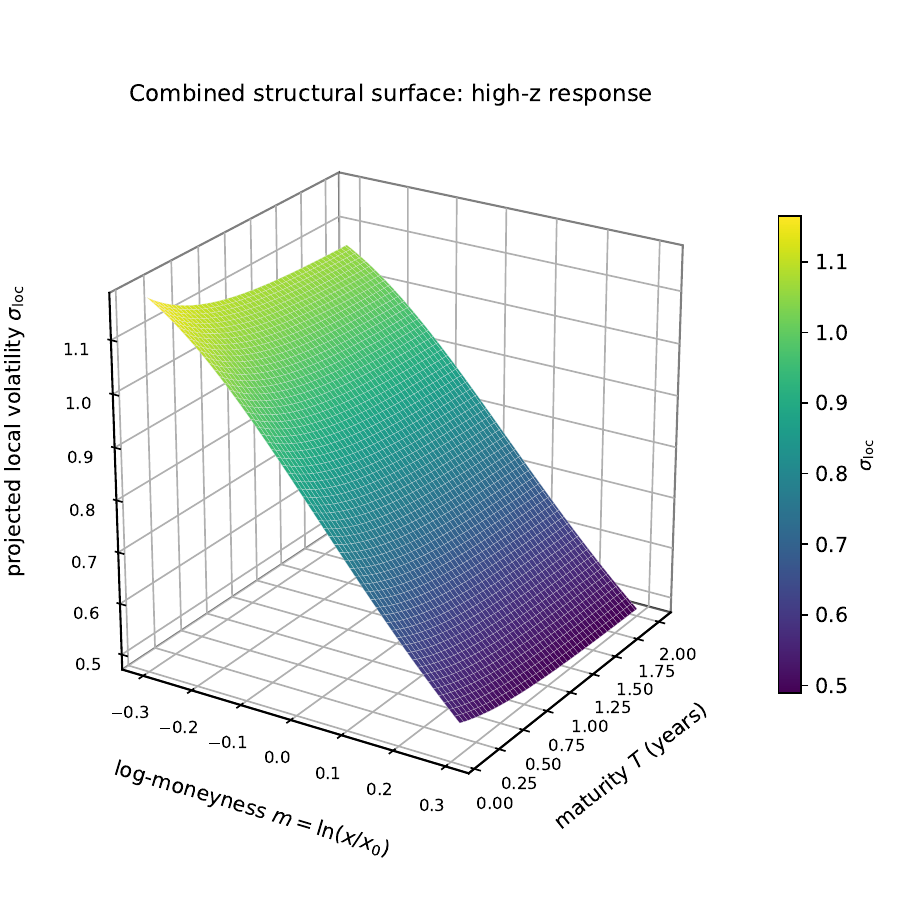}
    \textnormal{(a)}
\end{minipage}
\hfill
\begin{minipage}[t]{0.48\linewidth}
    \centering
    \includegraphics[width=\linewidth]{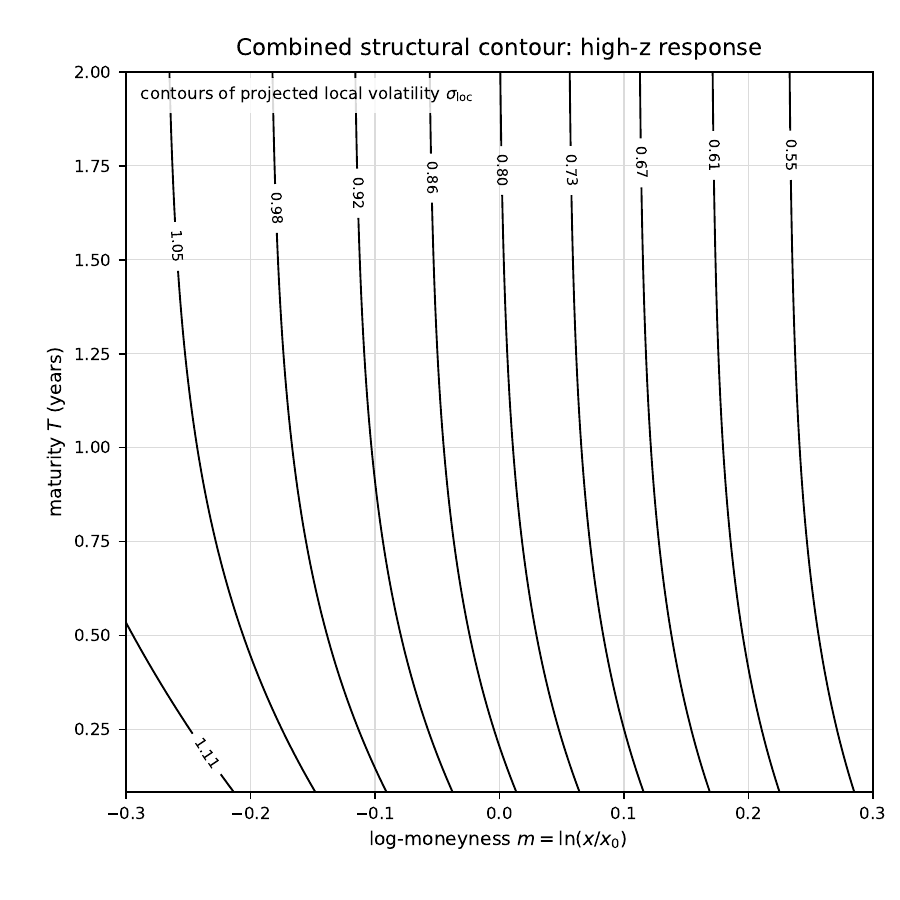}
    \textnormal{(b)}
\end{minipage}
\caption{Combined structural simulation in the high-$z$ response regime. (a) Projected structural volatility surface and (b) the corresponding contour representation. This profile combines the inverse-square liquidity channel, the response multiplier $K_\gamma^{(0)}z^{\gamma-1}$, and a stylised downside forcing profile. The two panels show the same state- and maturity-dependent morphology in complementary representations. We see that the three structural channels acting together produce the strongest downside asymmetry of the simulated cases. This is the closest of the examples to the intended phenomenological comparison with a calibrated volatility surface, while remaining a mechanism simulation rather than an empirical fit.}
\label{fig:sim-combined-high-z}
\end{figure}

\FloatBarrier

Taken together, the six paired surface--contour results separate the three structural channels before combining them. The liquidity-only cases show the inverse-square boundary response, the resilience-only comparison shows where the response factor begins to generate state dependence, and the combined cases show how these effects interact with the stylised forcing profile. Figure~\ref{fig:sim-combined-high-z} remains the most complete example considered here. 

Under the deterministic activity-clock benchmark and a chosen pricing measure $\Q$, the simulated operational kernel enters the calendar-time coefficient through Eq.~\eqref{eq:deterministic-localvol}. The numerical construction therefore illustrates the same ordering used in the analytic argument: reaction-boundary quantities first determine an operational variance kernel, after which the activity clock projects that kernel into calendar time.

This also gives the direct connection with the broader operational-time lattice construction in \cite{AngstmannGebbieBSMFromLattice2026}. There the finite-scale reaction-boundary variance enters the general local kernel before deterministic clock projection, whereas the present paper provides the detailed Green-function cumulant, finite-cutoff and asymptotic calculation from which that coefficient is obtained. The numerical surfaces above demonstrate this structural branch visually.

The simulations are therefore illustrative. They do not estimate $\cL(m)$, $\nu(m)$, $\Aeff(m)$, $\gamma(m)$ or $\alpha(T)$ from market data, and the plotted object should not be read as a calibrated or arbitrage-free Dupire local-volatility surface. Rather, the figures show which declared structural channels and their associated parameters can generate different state- and maturity-dependent surfaces. At moderate $z$, particularly for small $\gamma$, the leading high-$z$ approximation can differ appreciably from the crossover integral; quantitative applications should use the full integral or its finite-cutoff counterpart. The profiles used here separate state and maturity dependence, and in the combined high-$z$ case the forcing and resilience sensitivities enter through the single combination $A_{\mathrm{stress}}+(1-\gamma)\nu_{\mathrm{skew}}$. They cannot be identified separately from this surface alone.

\section{Discussion and limitations}

Here we provide a simple analytic bridge from microstructure response to local volatility. It uses the same Green-function response that explains concave impact \citep{toth2011anomalous,Mastromatteo2014, Donier2015}; but computes a second cumulant rather than a deterministic impact path. The microstructural input also sits in the broader order-flow-memory and latent-liquidity literature, where persistent signed flow, hidden-order splitting, slow liquidity digestion, and impact invariance are treated as structural market features rather than exogenous volatility factors \citep{LilloMikeFarmer2005,Bouchaud2004,bouchaud2009markets,Benzaquen2018}. For an admissible diffusion closure the deterministic projection gives a local-volatility coefficient in the usual pricing PDE. The non-unique-time extension leaves the operational kernel unchanged and moves model uncertainty into the projection layer. The construction should therefore be read as a structural closure rather than an empirical volatility model. 

What is derived is the form of a mesoscopic operational variance kernel once the local book slope, resilience, signed-forcing intensity, memory exponent, and operational coarse-graining scale are specified. What is assumed is local linearity of the latent book, frozen response coefficients over the increment window, locally stationary long-memory signed forcing, finite operational coarse graining, and a controlled low-frequency asymptotic regime. The key limitations then follow as: (i) the homogeneous limit is not a calibrated data model, (ii) the physical order-flow statistics are not identified directly with risk-neutral quantities, and (iii) it is not given that the finite-cutoff corrections are negligible in every empirical regime.

The operational scale $\Delta$ is part of the mesoscopic closure. It regularizes the singular short-scale response and should not be identified with the calendar-time maturity $T$ of a particular option contract. Within a single projected model, options across maturities should ordinarily be valued from one declared operational resolution rather than by resetting $\Delta$ contract by contract. The resilience time $\nu_u^{-1}$ is a response crossover scale, not a uniquely derived optimal choice of $\Delta$.

Superficially, because the frozen-coefficient boundary process of the latent limit order book produces subdiffusive price dynamics, one might conclude, by analogy with the propagator balance between persistent order flow and decaying impact, that long-memory forcing can overcome this limitation only at the precise balance \citep{bouchaud2009markets} \begin{equation} 
2\beta+\gamma=1, 
\label{eq:diffusive-balance} 
\end{equation} 
where the Abel response fixes $\beta=\tfrac12$. This would require $\gamma=0$ for diffusive scaling, since in the low-resilience regime the mean-squared boundary displacement scales as $\Delta^{1-\gamma}$, whereas the present construction assumes $0<\gamma<1$. However, the applicability of the scale-free balance in Eq.~\eqref{eq:diffusive-balance} depends fundamentally on the dimensionless response scale $z = \nu_u \Delta$, which separates two distinct asymptotic regimes:

First, in the \emph{low-$z$ regime} ($z \ll 1$), the coarse-graining scale $\Delta$ is much shorter than the resilience relaxation timescale $\tau_{\mathrm{res}} = \nu_u^{-1}$. Exponential tempering is negligible ($e^{-\nu \tau} \approx 1$), and the boundary responds through an untempered, scale-free Abel kernel with $\beta = \tfrac12$. In this regime, the balance in Eq.~\eqref{eq:diffusive-balance} serves as the canonical benchmark: it indicates that an untempered reaction boundary driven by persistent signed flow ($0 < \gamma < 1$) cannot achieve macroscopic diffusive scaling on its own, producing instead the subdiffusive mean-squared displacement $\Delta^{1-\gamma}$. This confirms that the stationary boundary displacement derived here describes transient fluctuations around a locally frozen interface---a tangent generator rather than the complete macroscopic price diffusion. In this controlled baseline (\ref{app:conditioning-wedge}), variation in resilience $\nu_u$ does not generate leading-order state dependence, and smile morphology is driven by liquidity depth $\cL(m)$ and forcing $A_{\mathrm{eff}}(m)$.

Second, in the \emph{high-$z$ regime} ($z \gg 1$), the coarse-graining scale $\Delta \gg \tau_{\mathrm{res}}$ extends far beyond the book's memory. In this extreme regime, the Bouchaud balance in Eq.~\eqref{eq:diffusive-balance} does not and cannot hold. The balance fundamentally presumes an untempered, scale-invariant power-law propagator extending indefinitely.\footnote{The balance in Eq.~\eqref{eq:diffusive-balance} applies to an untempered scale-free propagator, $G(\tau)\sim\tau^{-\beta}$. The present Abel response is exponentially tempered, so for fixed $\nu>0$ the balance can apply only before the resilience cutoff dominates; beyond that scale the stationary increment variance saturates to a constant, producing an asymptotic variance rate of order $\Delta^{-1}$.} For $z \gg 1$, finite resilience $\nu > 0$ breaks scale invariance: the exponential cutoff $e^{-\nu \tau}$ suppresses the contribution of old forcing to the boundary response, and boundary displacements mean-revert around the imbalance zero. Consequently, the increment variance saturates to a bounded constant ($\mathbb E[(\Delta X)^2]=\mathcal O(1)$), forcing the operational variance rate to decay as $a_u^{(\Delta)} \propto \Delta^{-1}$. High-$z$ is therefore a resilience-dominated, scale-broken regime where the state-dependent resilience profile $\nu(m)$ directly shapes the projected surface morphology (as illustrated in Figures~\ref{fig:sim-resilience-high-z} and~\ref{fig:sim-combined-high-z}), rather than a scale-free propagator balance.

The balance in Eq.~\eqref{eq:diffusive-balance} can thus describe an intermediate regime before the resilience cutoff dominates, but the exponentially tempered Abel response does not produce asymptotic diffusion at fixed $\nu>0$. Recovering diffusive long-scale prices requires additional dynamics beyond this stationary tangent process. An untempered propagator satisfying the appropriate balance with persistent order flow is one possible mechanism, but that extension is not derived here. Here the finite-resolution second cumulant supplies the variance coefficient for a proposed operational-time diffusion closure at a declared resolution $\Delta$, after which operational time must be projected separately into continuous calendar time. When time is non-unique, that projection need not be global or unique. The finite-scale variance identifies the coefficient of a proposed Markovian reduced-form representation; it does not establish weak convergence of the underlying reaction boundary to that diffusion, nor equality of their multi-time transition laws.

An intermediate diffusive regime is not excluded by the low- and high-$z$ asymptotics, although it cannot arise merely by varying $z=\nu_u\Delta$ in the present exponentially tempered Abel response. One possible mechanism is a broad distribution of latent-liquidity reassessment times \citep{BenzaquenBouchaud2018MultiTimescale}, represented by fractional reaction--diffusion or a subordinated liquidity clock with exponent $0<\theta<1$. In a scaling interval where the effective propagator behaves as $G_{\mathrm{eff}}(\tau)\propto\tau^{-\theta/2}$, persistent signed forcing with $C_m(\tau)\propto\tau^{-\gamma}$ satisfies an effective Bouchaud balance 
\begin{equation}
2\beta_{\mathrm{eff}}+\gamma=1
\label{eq:eff-diffusive-balance} 
\end{equation} 
when $\theta=1-\gamma$. \ref{app:collective-eft} outlines such a mechanism, one in which collective translation and the elimination of liquidity reassessment modes produce the effective propagator balance in Eq.~\eqref{eq:eff-diffusive-balance}, conditional on a prescribed operational-time reassessment spectrum. \ref{app:eft-matched-kernel} considers a candidate matched response kernel and the corresponding second cumulant. This suggests that a more general effective field theory of the limit-order book can be constructed.

This possibility requires additional dynamics beyond the stationary reaction-boundary calculation presented here. In particular, subordinating liquidity reassessment is not equivalent to subordinating the complete boundary process or to the separate operational-to-calendar-time projection used for pricing. Establishing an intermediate diffusive regime for the full moving boundary would require justifying the modified boundary Green function and its second cumulant, controlling a scaling window before the relevant relaxation cutoffs dominate, and verifying that the resulting projection remains adjoint-consistent. A finite interval of approximately linear variance growth would not, by itself, establish asymptotic diffusion or weak convergence to Brownian motion. This mechanism is consistent with the fractional latent-order-book construction of \citet{BenzaquenBouchaud2018Fractional}, in which heterogeneous liquidity relaxation compensates persistent order flow without requiring $\gamma=0$.

One might instead seek the intermediate diffusive regime entirely through stochastic subordination of the existing boundary trajectory. Such a clock can alter the calendar-time appearance of impact and order-flow correlations, but does not itself modify the operational Abel Green function or the reaction-front dynamics from which square-root metaorder impact is obtained \citet{AngstmannGebbieJSTAT2026}. For example, if the operational low-$z$ increment variance scales as $\Delta^{1-\gamma}$, independent inverse-stable subordination with index $0<\eta\leq1$ produces origin-based calendar-time variance scaling as $t^{\eta(1-\gamma)}$ within the corresponding scaling window, rather than linear variance growth for $0<\gamma<1$.\footnote{For an independent inverse $\eta$-stable clock
$E_t$ and a centred operational process with
$\mathbb{E}[(X(u)-X(0))^2]\simeq C u^{1-\gamma}$,
the inverse-stable moment identity gives
$\mathbb{E}[(X(E_t)-X(0))^2]
\simeq C\,\mathbb{E}[E_t^{1-\gamma}]
= C\frac{\Gamma(2-\gamma)}
{\Gamma(1+\eta(1-\gamma))}t^{\eta(1-\gamma)}$.
\citet{MeerschaertStraka2013}.} This does not exclude more general dependent or state-dependent clocks, but shows that external subordination alone does not supply the operational response exponent required by Eq.~\eqref{eq:eff-diffusive-balance}. The fractional liquidity mechanism considered above is different: it modifies the transport and reassessment dynamics that determine the reaction-boundary Green function, while the projection into calendar time remains a separate problem.

Here the resulting correspondence with local volatility is therefore conditional rather than automatic: a candidate projection must produce a common calendar-time pricing mechanism whose forward and backward operators satisfy the adjoint-real consistency condition. In this interpretation, $\Delta$ is not related to the maturity of an option, nor is it merely an arbitrary smoothing parameter. It is the operational resolution at which the event-time variance is represented within the chosen calendar-time closure. One therefore cannot equate the underlying mean-squared displacement with a weighted average of volatility to expiry or reset $\Delta$ separately for each maturity. The fundamental point is that anomalous operational-time scaling, clock projection, and calendar-time option valuation are distinct and necessary layers of the construction.

The scale-free dependence does not provide an intrinsic minimiser or distinguished value of $\Delta$, so the theory alone cannot select one. We instead interpret $\Delta$ as a model-resolution parameter and propose selecting a common value $\Delta_\ast$, independent of option maturity, by a projection-stability criterion. Specifically, we propose selecting $\Delta_\ast$ within the structural scale window
\begin{equation}
\tau_0\ll\Delta_\ast\ll\nu^{-1}
\label{eq:resolution-selection-window}
\end{equation}
subject to empirical tests of projection stability under modest changes of resolution. Approximate Markovianity and negligible predictive content of unresolved history are possible diagnostics, rather than consequences of the asymptotic calculation. The selected \(\Delta_*\) is then used for all maturities within the same pricing model. The theory identifies a structural range, but does not establish that any part of this range admits a diffusion closure or determine a unique universal resolution. At fixed physical coefficients the low-$z$ variance rate itself changes by a factor $2^{-\gamma}$ on doubling $\Delta$.

The role of an admissible Markovian closure is therefore analogous to that of local volatility in option pricing: it provides a reduced-form pricing representation whose coefficients may be calibrated to observed option prices, without requiring the representation to reproduce the full stochastic dynamics of the underlying market. The distinction here is that the functional form of the local variance is constrained by the operational reaction-boundary cumulant rather than inferred solely from option prices. Unlike conventional local-volatility calibration, the structural parameters may be estimated directly from order-book dynamics, while their risk-neutral counterparts must be determined or constrained separately for pricing.

The stationary spectral calculation and the state-conditioned closure involve distinct approximations. \ref{app:initialization-transients} shows that the omitted finite-initialization contribution decays in mean square once the elapsed operational time is large relative to the resilience time. \ref{app:conditioning-wedge} shows, within the frozen stationary Gaussian tangent system, that the unconditional spectral second moment differs relatively from the tangent conditional second moment by $\mathcal O((\nu_u\Delta)^{1-\gamma})$ for bounded standardized tangent states. Together with the finite-cutoff approximation, this identifies the controlled window
\begin{equation}
\nu_u u\gg1,\qquad \tau_0\ll\Delta\ll\nu_u^{-1},
\label{eq:controlled-window}
\end{equation}
which requires $\nu_u\tau_0\ll1$. This is an asymptotic compatibility regime, not a claim that cutoff, initialization, or conditioning corrections are negligible in every empirical setting.

A scalar projected coefficient is also not by itself a complete one-state pricing mechanism. The backward valuation and forward state-price evolution must arise from a common projected transition mechanism on compatible domains. Hidden activity may obstruct closure on the observed price state, even when a scalar coefficient reproduces selected marginal quantities. If no common one-state mechanism exists within the selected projection class, that projection must be rejected or the state description enlarged. This is a diagnostic of the proposed representation; it is not a general impossibility theorem for one-state projections and does not by itself prove financial-market incompleteness.

The relation to the broader non-unique-time discussion is specific. The present paper does not require a full operational-time lattice construction, a full meta-order impact model, or a multivariate correlation model\footnote{The single-book resilience rate $\nu$ is distinct from the cross-book spread-relaxation rate $\kappa$ in \citet{AngstmannGebbieEpps2026}; the pair-trader coupling source in that model is not identified with the signed point forcing $m(u)$ used here.}. It uses only the local linear response of the boundary, the forcing covariance, and the clock projection. The broader motivation for treating clock choice as a source of incompleteness is discussed in \citet{AngstmannGebbie2026NonUniqueTime}; but the derivation here is self-contained.

We do not assert that the physical signed-order-flow covariance is itself the risk-neutral covariance, nor that the finite-cutoff asymptotic prefactor is universal across assets or sampling schemes. Nor do we claim that every long-memory or fractional extension remains a scalar local-volatility model. In particular, the stochastic inverse-clock subordination described in \citet{AngstmannGebbieJSTAT2026} need not admit the one-state Markov diffusion used here as the deterministic-clock pricing benchmark. If the projected variance retains its own memory state, the Dupire coefficient is a marginal projection of richer variance dynamics rather than the primitive boundary variance itself. 

Instead, the paper identifies the structural objects that must be transformed, calibrated, or projected before the kernel can be used for pricing. In this sense the formula is most useful as a disciplined decomposition of local volatility into liquidity response, signed-flow memory, activity, and measure choice, and makes concrete the role of time non-uniqueness in discrete event-driven systems. 


\section{Conclusion}

In summary, what we have derived is a coarse-grained local variance coefficient from the linear response of a reaction boundary in a latent order book. The key result is Eq.~\eqref{eq:operational-result}. It expresses the finite-scale second cumulant per unit operational time in terms of directly interpretable structural quantities: the effective signed-forcing intensity $\Aeff$, the local liquidity slope $\cL$, the resilience $\nu$, the long-memory exponent $\gamma$, and the operational coarse-graining scale $\Delta$.

In simple physical terms, stronger signed forcing from sequences of buy and sell orders will increase the boundary variance, while a steeper latent book suppresses boundary motion through the inverse-square factor $1/\cL^2$. Long-memory order flow controls the scale dependence through $\Delta^{-\gamma}$, and the resilience of the order-book enters through the response function $\Fcal_\gamma(\nu\Delta)$, filtering the contribution of persistent forcing according to the ratio between the book relaxation scale and the measurement scale. 

Equation~\eqref{eq:deterministic-localvol} then gives the deterministic-clock benchmark: for an admissible diffusion closure, calendar-time local variance is the operational variance rate multiplied by the activity rate $\alpha(t)$. Thus activity does not alter the Green-function response kernel; it determines how rapidly operational fluctuations are sampled in calendar time. 

The coefficient is therefore mesoscopic rather than an unregularized microscopic variance. Its operational resolution $\Delta$ is distinct from contractual maturity, and its stationary tangent interpretation is controlled only within the scale window identified in Eq.~\eqref{eq:controlled-window}. The resilience time $\nu^{-1}$ bounds the low-$z$ regime as a response crossover scale; it is not selected here as a universal optimal resolution. What is important here is that the frozen boundary fluctuations are not themselves diffusive. At low $z$ the increment variance grows as $\Delta^{1-\gamma}$, while at high $z$ it saturates. The corresponding variance rates thus depend on the scale at which the boundary is observed. The high-$z$ surfaces illustrate the resilience dependence of this variance rate rather than the emergence of a diffusion.

Finally, Eq.~\eqref{eq:homogeneous-result} gives the homogeneous closed form, where the structural quantities are frozen. This limit is not meant to be a complete empirical model, but it is the clearest summary of the mechanism: activity amplifies variance, signed forcing supplies fluctuations, liquidity absorbs them quadratically, resilience damps persistent memory, and the memory exponent determines how the observed variance changes with coarse graining. The non-unique-time extension preserves this operational cumulant and moves unresolved clock or projection ambiguity into the map from operational to calendar time; whether that ambiguity becomes financial-market incompleteness depends on the completed market model.

The mathematical message is simple. The Green-function calculation determines an operational variance kernel under explicit local and asymptotic assumptions. The local-volatility pricing interpretation then requires a diffusion closure, a clock and a pricing measure. The term closure\footnote{The term ``closure'' comes from a ``closed system of equations'' where the evolution
of retained variables is determined without introducing new ones. The problem arose in kinetic theory
and statistical turbulence, where the equations for successive
moments form an infinite hierarchy; see H.~Grad,
\textit{Commun.\ Pure Appl.\ Math.} \textbf{2}, 331--407 (1949).
The terminology was already established by 1959, when Kraichnan
used ``closure'' to describe the reduction of such a hierarchy
to a determinate set of equations; see R.~H.~Kraichnan,
\textit{J.\ Fluid Mech.} \textbf{5}, 497--543 (1959),
specifically p.~499. A similar closure problem arises in relativistic kinetic theory and cosmology, where the choice of frame and assumptions concerning the kinetic moments enter the
construction of a closed system of equations; see
G.~F.~R.~Ellis, D.~R.~Matravers and R.~Treciokas,
\textit{Ann.\ Phys.} \textbf{150}, 455--486 (1983).} is used here in the sense of coarse graining in statistical physics \citep{Zwanzig1961Memory,Mori1965Transport}: we seek to describe the evolution of the retained price state without explicitly evolving the unresolved degrees of freedom of the order book. The calculation of a second cumulant is part of this reduction, but does not by itself ``close'' the dynamics. The scale window in which we can calculate the tangent cumulant need not be one in which a diffusion closure is admissible. In particular, the restoring drift of the frozen boundary in \ref{app:conditioning-wedge} is not the pricing drift. The collective response considered in Appendices~\ref{app:collective-eft} and~\ref{app:eft-matched-kernel} suggests how diffusive scaling might emerge through the reassessment of liquidity, but this is conditional on the reassessment spectrum and the matching of the response kernels. It is not a diffusion limit of the frozen boundary. If the clock is unique and deterministic and the diffusion closure is admissible, the result reduces to an activity-rescaled local-volatility coefficient. If the clock is non-unique, the admissible projections must preserve the adjoint relation between valuation and density evolution. This is where the microstructure calculation meets mathematical finance: the order book supplies the operational cumulant, but pricing requires an adjoint-real calendar-time representation.

\section*{Code and reproducibility}

The simulations and original figure outputs are archived in the versioned v1.0.0 reproducibility release \cite{GebbieAngstmann2026RBVSoftware}.
The finite-rate collective benchmarks discussed in \ref{app:collective-eft} are documented in the separate v2.0.0 computational supplement \citep{reactiondata2026}. The subsequent effective-response calculations and associated figure sources and reproducibility scripts are maintained in the versioned repository (v2.1.0) \citep{reactiondata2026v21}; these supplementary materials are not included in the manuscript source bundle.

\section*{Acknowledgments}
We are grateful for thoughtful community feedback and critique that have helped us improve the paper.

\section*{AI Disclaimer}

The authors used generative AI (OpenAI ChatGPT, GPT-6 family including GPT-6.2, and Anthropic Claude, Opus 5) for language editing, typesetting and consistency and mathematical checks. ChatGPT also assisted with project staging, specified mathematical derivations and derivation checking. OpenAI Codex assisted with the reproducibility materials. The authors reviewed and edited the AI-assisted outputs and take full responsibility for the mathematical arguments, scientific claims and manuscript content.

\bibliographystyle{elsarticle-harv}
\bibliography{AdjointProjection-arXiv-v3.0.5}

\appendix
\section{Finite-cutoff spectral form}
\label{app:finite-cutoff}
The asymptotic closure in Eq.~\eqref{eq:operational-result} uses the low-frequency, zero-cutoff approximation to the regularised spectral cumulant. This appendix records the finite-cutoff object from which that approximation is taken.

With
\begin{equation}
 g_{\nu,D}^{(\tau_0)}(r)=\frac{\exp[-\nu(r+\tau_0)]}{\sqrt{4\pi D(r+\tau_0)}}\mathbf 1_{r\ge0},
\end{equation}
the exact filtered coefficient at operational scale $\Delta$ is
\begin{equation}
 a_{u,\tau_0}^{(\Delta)}=
 \frac{1}{\Delta\cL_u^2}\frac{1}{2\pi}
 \int_{-\infty}^{\infty}4\sin^2\left(\frac{\omega\Delta}{2}\right)
 \left|\wh g_{\nu,D}^{(\tau_0)}(\omega)\right|^2
 S_m^{(\tau_0)}(\omega)\dd\omega,
 \label{eq:finite-cutoff-cumulant}
\end{equation}
where $S_m^{(\tau_0)}$ is the spectrum associated with
\begin{equation}
 C_m^{(\tau_0)}(\tau)=A_m(|\tau|+\tau_0)^{-\gamma}.
\end{equation}
The zero-cutoff calculation replaces
\begin{equation}
 \left|\wh g_{\nu,D}^{(\tau_0)}(\omega)\right|^2
 \quad\text{by}\quad
 \frac{1}{4D\sqrt{\nu^2+\omega^2}}
\end{equation}
and
\begin{equation}
 S_m^{(\tau_0)}(\omega)
 \quad\text{by}\quad
 A_m C_\gamma |\omega|^{\gamma-1}
\end{equation}
in the low-frequency regime relevant to $\Delta/\tau_0\gg1$. Under the change of variables $\xi=\omega\Delta$, these approximations give Eq.~\eqref{eq:operational-result}. Finite $\tau_0$ modifies the prefactor and the crossover at large $\nu\Delta$, but it does not change the leading long-memory scaling in the controlled regime where the increment scale is separated from the cutoff scale. This is deliberately stated at the spectral level. A calibrated finite-cutoff model would retain Eq.~\eqref{eq:finite-cutoff-cumulant} directly rather than replacing it by the asymptotic response function $\Fcal_\gamma$.

For the high-$z$ zero-cutoff prefactor, the response-selected frequency scale is of order $\nu$, so consistency with the low-frequency cutoff replacement additionally requires $\nu\tau_0\ll1$.

\section{Adjoint consistency}
\label{app:adjoint-consistency}

In Proposition \ref{prop:adjoint-real} we provide the clock consistency reality condition and then provide a short outline for the proof strategy. Here we prove the consistency assertion using the proposed kernel-duality argument for Markov transition mechanisms and their infinitesimal generators \cite{EthierKurtz1986,Dynkin1965MarkovProcessesI,Feller1971Volume2}. We do not prove existence or uniqueness of a pricing measure, but identify the minimal coherence condition required once an operational-time variance kernel is projected into calendar time.  A projected clock can be used as a one-state pricing representation only when the backward valuation operator for claims and the forward state-price operator for laws are induced by the same discounted, or killed, state-price kernel. This serves as the formal version of the adjoint-real consistency requirement used in the main text; and makes more explicit the operational-lattice and non-unique-time interpretations in \cite{AngstmannGebbieBSMFromLattice2026,AngstmannGebbie2026NonUniqueTime}.

\begin{definition}[Projected one-state pricing mechanism]
Fix a log-price domain $\Sdom\subseteq\mathbb R$, a calendar interval $[0,T]$, a pricing measure $\Q$, and a class $\mathfrak U$ of operational-to-calendar projections for an operational variance kernel $\Xiop(S,u;\Delta)$.  For $U\in\mathfrak U$, a projected one-state pricing mechanism is a two-parameter discounted, or killed, state-price kernel
\[
 \Kern^U_{s,t}(S,\dd Y),\qquad S,Y\in\Sdom,\qquad 0\leq s\leq t\leq T,
\]
on the log-price domain $\Sdom$.  Its backward action on claims is
\[
 (P^U_{s,t}V)(S)=\int_{\Sdom}V(Y)\Kern^U_{s,t}(S,\dd Y),
\]
and its forward action on state-price measures is the dual action of the same kernel.  Thus, whenever state-price densities exist with respect to the log-price coordinate $S$,
\[
 \PairS{p}{V}:=\int_{\Sdom}V(S)p(S)\dd S,
 \qquad
 \PairS{p}{P^U_{s,t}V}=\PairS{(P^U_{s,t})^*p}{V}.
\]
\end{definition}

\begin{definition}[Adjoint-real projection]
For the chosen log-price domain, information set, operator domains, boundary convention, and discounting/killing convention, define
\begin{align}
 \Ad(\mathfrak U)
 :=
 \left\{
 U\in\mathfrak U:
 \PairS{p}{\Bop_t^U V}=\PairS{\Gop_t^U p}{V} \right. \nonumber \\
 \left. \text{ for all admissible } p,V
 \right\}.
\end{align}
Here $\Bop_t^U$ is the backward claim generator and $\Gop_t^U$ is the forward state-price generator associated with the same projected mechanism.
\end{definition}

\begin{theorem}[Adjoint consistency as a necessary coherence constraint]
\label{thm:adjoint-consistency-v12}
Let $U\in\mathfrak U$.  Suppose that $U$ defines a coherent one-state projected pricing mechanism on the information set $\Ical_t^S$ through a single discounted, or killed, state-price kernel $\Kern^U_{s,t}$.  Suppose further that, on compatible domains $\Dclaim(t)$ and $\Ddens(t)$, the paired short-step expansions
\[
 P^U_{t,t+h}V=V+h\Bop^U_tV+o(h),\qquad
 (P^U_{t,t+h})^*p=p+h\Gop^U_tp+o(h)
\]
hold after pairing with admissible tests.  Then, for every admissible state-price density $p\in\Ddens(t)$ and every admissible claim or test function $V\in\Dclaim(t)$,
\begin{equation}
 \PairS{p}{\Bop^U_tV}=\PairS{\Gop^U_tp}{V}.
 \label{eq:adjoint-consistency-v12}
\end{equation}
Consequently $U\in\Ad(\mathfrak U)$ for the selected log-price domain, information set, and boundary/discount convention.  Hence a non-unique operational-to-calendar projection can serve as a coherent one-state pricing representation only if its backward claim operator and forward state-price operator are adjoints of the same projected mechanism.

In particular, if a projection supplies only a scalar coefficient
\[
 \Sigma_U^2(S,t)=\Pproj_U[\Xiop(S,u;\Delta);t]
\]
for a plausible backward pricing equation, but does not supply the corresponding forward state-price evolution through the same one-state mechanism, then it is a coefficient projection rather than a coherent one-state pricing model.
\end{theorem}

\begin{proof}
For a fixed projection $U$, write $P_{s,t}=P^U_{s,t}$.  Since the backward valuation of claims and the forward evolution of state-price laws are induced by the same discounted or killed kernel, the finite-step duality identity holds:
\[
 \PairS{p}{P_{s,t}V}=\PairS{P_{s,t}^*p}{V}.
\]
This is simply Fubini--Tonelli applied to
\[
 (P_{s,t}V)(S)=\int_{\Sdom}V(Y)\Kern^U_{s,t}(S,\dd Y),
\]
with the common kernel moved from the claim side to the state-price side.

Apply the identity over the short interval $[t,t+h]$.  Using the paired generator expansions gives
\[
 \PairS{p}{V+h\Bop^U_tV+o(h)}
 =
 \PairS{p+h\Gop^U_tp+o(h)}{V}.
\]
Cancel $\PairS{p}{V}$, divide by $h$, and let $h\downarrow0$.  The paired remainders vanish by assumption, yielding
\[
 \PairS{p}{\Bop^U_tV}=\PairS{\Gop^U_tp}{V}.
\]
This is precisely the adjoint-consistency condition in \eqref{eq:adjoint-consistency-v12}.  Therefore any coherent one-state projected pricing mechanism must belong to $\Ad(\mathfrak U)$.

The coefficient-only statement follows because the preceding argument cannot be started from a scalar variance map alone.  The proof requires a common finite-step state-price mechanism and its dual action on state-price laws.  A projected scalar variance may therefore be a useful local coefficient, but it is not by itself a coherent one-state pricing model.  This proves the necessary condition.
\end{proof}

\begin{remark}[Deterministic-clock benchmark]
If $U$ is deterministic and absolutely continuous with $\dot U(t)=\alpha(t)$, and if the risk-neutral log-price variance is
\[
 a_U^\Q(S,t)=\alpha(t)\Xiop^\Q(S,U(t);\Delta),
\]
then the killed log-price pricing generator may be written
\[
 \Bop^U_t f=b_U^\Q f_S+\frac12a_U^\Q f_{SS}-rf,
 \qquad
 b_U^\Q=r-q-\frac12a_U^\Q.
\]
Its formal adjoint is
\[
 \Gop^U_t p=-\partial_S(b_U^\Q p)+\frac12\partial_{SS}(a_U^\Q p)-rp,
\]
assuming compact support or boundary conditions that remove boundary terms.  In price coordinates this gives the corresponding killed local-volatility pricing equation
\[
 V_t+(r-q)xV_x+\frac12\alpha(t)\Xiop^\Q(\log x,U(t);\Delta)x^2V_{xx}-rV=0.
\]
\end{remark}

\begin{remark}[Hidden clocks and state enlargement]
If clock increments depend on hidden information not measurable with respect to $\Ical_t^S$, then a scalar projected variance may match a marginal variance or selected option prices without defining the same full one-state pricing semigroup as the hidden-clock model.  This is not an impossibility result for Markovian projection.  It says only that coherent one-state pricing requires a common one-state mechanism.  If the hidden state is relevant, the state should be enlarged, for example to $(S,Z)$ or to a sufficient filter state $(S,\pi)$.
\end{remark}

\begin{remark}[Kernel notation]
To distinguish the primitive operational transition kernel from the projected calendar-time state-price kernel we note that the kernel hierarchy is:
\[
K_{\Delta u}(S,\dd Y;u)
 \;\longrightarrow\;
K^{\Q}_{\Delta u}(S,\dd Y;u)
 \;\longrightarrow\;
{\mathcal K^U_{s,t}(S,\dd Y)} .
\]
Here $K_{\Delta u}$ is reserved for the local operational-time transitions, and $\mathcal K^U_{s,t}$ for the calendar-time projected, discounted or killed, state-price kernel used in the adjoint-consistency argument.
\end{remark}

In summary: the appendix proves the consistency condition underlying Proposition~1: a non-unique time projection is admissible as a one-state pricing representation only when valuation and state-price evolution are generated by the same discounted state-price kernel.  The result is a necessary condition, not a sufficiency theorem.  It supports the use of $\Ad(\mathfrak U)$ as an admissibility set, while leaving measure selection, empirical calibration, market completeness, and possible state enlargement as separate modelling choices.

\section{Transients and Convergence}
\label{app:initialization-transients}

This appendix quantifies the relation between the finite-initialization response in Eq.~\eqref{eq:linear-response} and the stationary spectral cumulant in Eq.~\eqref{eq:spectral-cumulant}. The purpose is not to replace the finite-start dynamics by assumption, but to identify a regime in which the stationary representation is an accurate asymptotic closure.

Let the boundary displacement initialized at operational time $u=0$ be
\begin{equation}
Y_0(u)=\frac{1}{\cL_u}\int_0^u g_{\nu,D}(u-s)m(s)\dd s,
\label{eq:finite-start-response}
\end{equation}
and let the corresponding stationary displacement over the infinite past be
\begin{equation}
Y_\infty(u)=\frac{1}{\cL_u}\int_{-\infty}^u g_{\nu,D}(u-s)m(s)\dd s.
\label{eq:stationary-response}
\end{equation}
The response coefficients are frozen over the increment window, as in the main derivation. For $\Delta>0$, the finite-start increment decomposes as
\begin{equation}
\Delta_\Delta Y_0(u)=\Delta_\Delta Y_\infty(u)-T(u,\Delta),
\label{eq:finite-stationary-decomposition}
\end{equation}
where
\begin{equation}
T(u,\Delta)=\frac{1}{\cL_u}\int_{-\infty}^0
\left[g_{\nu,D}(u+\Delta-s)-g_{\nu,D}(u-s)\right]m(s)\dd s.
\label{eq:initialization-transient}
\end{equation}

For the damped Abel kernel in Eq.~\eqref{eq:abel-kernel}, set $r=-s\geq0$. Since $e^{-\nu(u-s)}=e^{-\nu u}e^{-\nu r}$,
\begin{align}
&\left|g_{\nu,D}(u+\Delta+r)-g_{\nu,D}(u+r)\right| \nonumber\\
&\quad\leq \frac{e^{-\nu u}e^{-\nu r}}{\sqrt{4\pi D}}
\left[\frac{e^{-\nu\Delta}}{\sqrt{u+\Delta+r}}+\frac{1}{\sqrt{u+r}}\right]
\leq \frac{2e^{-\nu u}e^{-\nu r}}{\sqrt{4\pi D(u+r)}}.
\label{eq:transient-kernel-bound}
\end{align}
Assume, consistently with Eq.~\eqref{eq:forcing-cov}, that
\begin{equation}
|C_m(\tau)|\leq A_m(\tau_0+|\tau|)^{-\gamma},\qquad 0<\gamma<1.
\label{eq:transient-covariance-bound}
\end{equation}
Then
\begin{align}
\Var[T(u,\Delta)]
&\leq \frac{A_m e^{-2\nu u}}{\pi D\cL_u^2}
\int_0^\infty\int_0^\infty
\frac{e^{-\nu(r+r')}(\tau_0+|r-r'|)^{-\gamma}}
{\sqrt{(u+r)(u+r')}}\dd r\dd r'.
\label{eq:transient-variance-bound}
\end{align}
The double integral is finite because the exponential factors control the tails; the regularized covariance also removes the diagonal singularity. Since $[(u+r)(u+r')]^{-1/2}\leq u^{-1}$ for $u>0$,
\begin{equation}
\Var[T(u,\Delta)]=\mathcal O\!\left(\frac{e^{-2\nu u}}{u}\right),
\qquad
\|T(u,\Delta)\|_2=\mathcal O\!\left(\frac{e^{-\nu u}}{\sqrt{u}}\right).
\label{eq:transient-decay}
\end{equation}
Consequently, $\Delta_\Delta Y_0(u)\to\Delta_\Delta Y_\infty(u)$ in mean square. More precisely,
\begin{align}
&\left|\Var[\Delta_\Delta Y_0(u)]-\Var[\Delta_\Delta Y_\infty(u)]\right| \nonumber\\
&\quad\leq \Var[T(u,\Delta)]
+2\sqrt{\Var[\Delta_\Delta Y_\infty(u)]\Var[T(u,\Delta)]}.
\label{eq:cumulant-convergence-bound}
\end{align}
Thus, when the stationary increment variance is finite, the finite-start cumulant converges exponentially, up to an algebraic prefactor, to the stationary cumulant in Eq.~\eqref{eq:spectral-cumulant}. The bound does not identify the finite-start and stationary processes at short elapsed times; it quantifies the decay of the omitted pre-initialization contribution for $\nu u\gg1$.

\section{Closures and Conditioning Bounds}
\label{app:conditioning-wedge}

This appendix compares the unconditional stationary increment cumulant with a tangent-state conditional second moment. The calculation is made within the locally frozen Gaussian linear system and is a closure result, not a claim that the full order-flow dynamics are globally Gaussian.

Let
\begin{equation}
X_\Delta(u)=\Delta_\Delta Y(u)=Y(u+\Delta)-Y(u),
\end{equation}
where $Y$ is a centred, jointly stationary Gaussian process with autocovariance $\Gamma_Y(\tau)=\E[Y(u+\tau)Y(u)]$. Define
\begin{equation}
V_\Delta=\Var[X_\Delta(u)],\qquad V_Y=\Var[Y(u)]=\Gamma_Y(0),
\end{equation}
and assume $0<V_Y<\infty$. Gaussian conditioning gives
\begin{equation}
\Var[X_\Delta(u)\mid Y(u)]=V_\Delta-
\frac{\operatorname{Cov}(X_\Delta(u),Y(u))^2}{V_Y}.
\label{eq:gaussian-conditional-variance}
\end{equation}
Stationarity implies
\begin{equation}
\operatorname{Cov}(X_\Delta(u),Y(u))=\Gamma_Y(\Delta)-\Gamma_Y(0)=-\frac12V_\Delta,
\label{eq:increment-state-covariance}
\end{equation}
and hence
\begin{equation}
\Var[X_\Delta(u)\mid Y(u)]=V_\Delta\left(1-\frac{V_\Delta}{4V_Y}\right).
\label{eq:conditional-increment-variance}
\end{equation}
The conditional mean is
\begin{equation}
\E[X_\Delta(u)\mid Y(u)=y]= -\frac{V_\Delta}{2V_Y}y,
\label{eq:conditional-increment-mean}
\end{equation}
so the conditional second moment relevant to the transport closure is
\begin{equation}
\E[X_\Delta(u)^2\mid Y(u)=y]
=V_\Delta+\frac{V_\Delta^2}{4V_Y}\left(\frac{y^2}{V_Y}-1\right).
\label{eq:conditional-second-moment}
\end{equation}

Within the same zero-cutoff long-memory approximation used in Sec.~\ref{sec:asymptotic-closure}, the low-$z$ increment variance scales as $V_\Delta\propto\Delta^{1-\gamma}$, while the stationary level variance scales as $V_Y\propto\nu^{\gamma-1}$. Therefore
\begin{equation}
\frac{V_\Delta}{V_Y}=\mathcal O\!\left((\nu\Delta)^{1-\gamma}\right)
=\mathcal O\!\left(z^{1-\gamma}\right),\qquad z=\nu\Delta\ll1.
\label{eq:conditioning-ratio}
\end{equation}
For tangent states satisfying $y^2/V_Y=\mathcal O(1)$, Eq.~\eqref{eq:conditional-second-moment} gives
\begin{equation}
\E[X_\Delta(u)^2\mid Y(u)=y]
=V_\Delta\left[1+\mathcal O\!\left(z^{1-\gamma}\right)\right].
\label{eq:conditional-second-moment-asymptotic}
\end{equation}

Accordingly, the mesoscopic tangent coefficient is
\begin{equation}
a_u^{(\Delta)}(S,u)=\frac{1}{\Delta}
\E^{(S,u)}\!\left[(\Delta_\Delta Y_{S,u})^2\right],
\label{eq:tangent-state-closure}
\end{equation}
where $\E^{(S,u)}$ denotes expectation under the centred Gaussian tangent system whose parameters
\[
 (\cL_u(S),\nu_u(S),\Aeff(S),\gamma(S))
\]
are frozen at $(S,u)$. Since $Y(u)=S(u)-S_0$, conditioning on $S(u)=S$ corresponds to $y=S-S_0$ in Eq.~\eqref{eq:conditional-second-moment}. The unconditional stationary spectral cumulant is therefore the leading tangent-state second moment in the low-$z$ regime for bounded standardized tangent states. This conclusion is local to the frozen Gaussian closure and does not assert that the full forcing process is globally Gaussian or Markovian.

\section{An effective diffusive response}
\label{app:collective-eft}

The balance in Eq.~\eqref{eq:eff-diffusive-balance} can be understood through an effective field theory that distinguishes local reaction-boundary fluctuations from collective translation of the liquidity field. In operational time $u$, primary signed forcing with $C_m(\tau)\sim A_m|\tau|^{-\gamma}$ has accumulated variance $V_M\propto\Delta^{2-\gamma}$; the frozen Abel response of Eq.~\eqref{eq:operational-result} instead gives $V_f\simeq(\Aeff/\cL_u^2)\Delta^{1-\gamma}\Fcal_\gamma(z)$, $z=\nu_u\Delta$. For fixed coefficients in the stationary zero-cutoff model,
\begin{equation}
 0<\frac{\partial\log V_f}{\partial\log\Delta}
 =1-\gamma+\frac{z\Fcal_\gamma'(z)}{\Fcal_\gamma(z)}<1-\gamma,\qquad z>0.
 \label{eq:eft-frozen-bound}
\end{equation}
Neither the subdiffusive Abel regime nor its crossover to saturation is diffusive. At positive microscopic cutoff an isolated unit-slope crossing is possible, but is not a diffusive interval.

We first identify the collective coordinate associated with translation of the liquidity field. Write $\Phi=T_{q_c}(\Phi_*+\eta_\perp)$ with $T_{q_c}=\exp(-q_c\partial_x)$ and $\eta_\perp$ orthogonal to the translation tangent $\psi=-\partial_x\Phi_*$. The Cartan form is $T_{q_c}^{-1}dT_{q_c}=-\partial_x\,dq_c$ for differentiable paths. In a translation-invariant moving-deposition model the linearised spectrum is $\{0\}\cup(-\infty,-\nu_u]$: translation is neutral and shape relaxation is gapped. The zero is $S=y=q_c+\xi$, where $\xi\simeq\eta_\perp(0)/\cL_u$. On the invariant aligned manifold ($\xi=0$), bare translation obeys $\dot q_c=\kappa_*m$, with mobility $\kappa_*>0$. Neutrality alone therefore gives superdiffusion, not diffusion; the moving and frozen response operators are distinct.

We next introduce a positive measure of internal liquidity reassessment rates $d\mu(\lambda)$ and auxiliary states $R_\lambda$, with $v=\dot q_c$:
\begin{equation}
 \begin{aligned}
 \dot R_\lambda&=-\lambda R_\lambda+v,\\
 v&=\kappa_* m-g_R\int R_\lambda\,d\mu(\lambda),\quad g_R>0.
 \end{aligned}
 \label{eq:eft-reassessment}
\end{equation}
Eliminating these modes gives a retarded feedback kernel $k(\tau)=\int e^{-\lambda\tau}d\mu(\lambda)$ and its Laplace transform $K(s)=\int d\mu(\lambda)/(s+\lambda)$. The zero-history velocity and displacement transfers are
\begin{equation}
 H(s)=\frac{\kappa_*}{1+g_R K(s)},\qquad
 \widehat G_{\mathrm{eff}}(s)=\frac{H(s)}{s}.
 \label{eq:eft-transfer}
\end{equation}
Here $m$ is primary rather than already compensated net flow. Stationary spectra require stationary preparation. Exact shell elimination at an operational rate cutoff $\Lambda$ preserves the transfer when histories are retained. For $\rho(\lambda)=(\sin\pi b/\pi)\lambda^{-b}$, $0<b<1$, the rescaled fast kernel $\Lambda^b K_>(\Lambda x;\Lambda)$ is a fixed function in the ideal continuum. This Wilsonian construction retains a prescribed exponent; it does not select $b$ kinetically.

We now consider the compensated response in the ideal continuum. The rate spectrum gives $K(s)=s^{-b}$. The response is the Mittag--Leffler function $G_{\mathrm{eff}}(\tau)=\kappa_* E_b(-g_R\tau^b)$ with tail proportional to $\tau^{-b}$. Thus $\beta_{\mathrm{eff}}=b$ in the ideal scaling regime, not necessarily at finite lags. With the convention of Eq.~\eqref{eq:fourier-convention},
\begin{equation}
 \begin{aligned}
 S_v(\omega)&=|H(\ii\omega)|^2S_m(\omega)
 \sim A_m C_\gamma(\kappa_*/g_R)^2|\omega|^{2b+\gamma-1},\\
 V_c(\Delta)&=\frac{2}{\pi}\int_0^\infty
 \frac{1-\cos(\omega\Delta)}{\omega^2}S_v(\omega)\,d\omega
 \ \propto\ \Delta^{2-2b-\gamma},
 \end{aligned}
 \label{eq:eft-variance}
\end{equation}
under spectral domination and $0<2-2b-\gamma<2$. The balance of \citet{Bouchaud2004} is then recovered as
\begin{equation}
 2\beta_{\mathrm{eff}}+\gamma=1,\qquad
 \frac{V_c(\Delta)}{\Delta}\longrightarrow
 A_m C_\gamma(\kappa_*/g_R)^2
 \label{eq:eft-balance}
\end{equation}
in the ideal large-lag limit. This matches the fractional liquidity scaling of \citet{BenzaquenBouchaud2018Fractional} ($b=\theta/2$) and the multi-timescale construction of \citet{BenzaquenBouchaud2018MultiTimescale} ($b=1/2-\alpha_{\mathrm{BB}}$, $2\alpha_{\mathrm{BB}}=\gamma$ at balance, for frequency-independent turnover). The latter retains forcing-age dependence. These are exponent correspondences, not identities of response operators, rate measures or increment laws.

\begin{remark}[Conditional collective scaling]
Equation~\eqref{eq:eft-balance} establishes linear variance growth in the ideal collective scaling limit under the stated spectral assumptions. It neither selects $b$ endogenously nor proves a diffusion limit.
\end{remark}

For a finite spectrum of reassessment rates, the collective scaling regime is necessarily restricted. For physical reassessment rates in $[\lambda_{\min},\lambda_{\max}]$, a nominal collective window is
\begin{equation}
 \max(\tau_0,\lambda_{\max}^{-1},\nu_u^{-1},g_R^{-1/b})
 \ll\Delta\ll\lambda_{\min}^{-1},
 \label{eq:eft-window}
\end{equation}
where the resilience scale is a sufficient separation condition for shape elimination, not required under exact alignment. This hierarchy alone does not control out-of-band power. A positive rate floor gives $H(0)>0$ and restores superdiffusive accumulation for persistent scalar input at long lags. Two prescribed aligned finite-rate benchmarks support approximately linear intermediate variance, but stricter envelope and offset tests, nonlinear diffusion and independent selection of $b$ remain unresolved \citep{reactiondata2026}. Off alignment $V_S=V_c+V_\xi+2C_{c\xi}$, with signed increment covariance $C_{c\xi}$ not removed by spatial orthogonality.

\begin{remark}[Distinct scaling windows]
The sufficient collective window differs from Eq.~\eqref{eq:controlled-window}. Neither window establishes simultaneous nonzero Abel subdiffusion and total-root diffusion.
\end{remark}

\section{An effective response kernel}
\label{app:eft-matched-kernel}

The collective construction in \ref{app:collective-eft} suggests extending the frozen response rather than replacing its cumulant calculation. With the local coefficients held fixed, include the inverse slope in the zero-cutoff displacement kernel $G_f(\tau)=g_{\nu,D}(\tau)/\cL_u$. A candidate common-input response is
\begin{equation}
 \widehat G_{\mathrm{EFT}}(s)\simeq
 \underbrace{\frac{1}{2\cL_u\sqrt{D_u}\sqrt{s+\nu_u}}}_{\widehat G_f(s)}
 +\underbrace{\frac{\kappa_*}{s[1+g_R K(s)]}}_{\widehat G_{\mathrm{eff}}(s)},
 \label{eq:eft-matched-candidate}
\end{equation}
where $K(s)=\int d\mu(\lambda)/(s+\lambda)$. The first term is the massive local response; the second is a dressed neutral translation mode. In the ideal homogeneous continuum $K(s)=s^{-b}$, the latter is $\kappa_* E_b(-g_R\tau^b)$ in operational time. In the feedback-dominated scaling region,
\begin{equation}
 G_{\mathrm{EFT}}(\tau)\simeq
 \frac{e^{-\nu_u\tau}}{\cL_u\sqrt{4\pi D_u\tau}}
 +\frac{\kappa_*}{g_R\Gamma(1-b)}\tau^{-b},
 \label{eq:eft-matched-tail}
\end{equation}
where $b=\beta_{\mathrm{eff}}=(1-\gamma)/2$. The tail approximation applies only where the collective response has entered its scaling regime. With separated scales and suitable amplitudes, the short-lag Abel response gives way to a slowly decaying collective contribution as resilience suppresses the local response.

The linearised moving-deposition model provides an explicit test of the proposed matching. Put $a=\sqrt{\nu_u/D_u}$ and $k_s=\sqrt{(s+\nu_u)/D_u}$. For unit-integral point forcing, its undressed root transfer is exactly
\begin{equation}
 \widehat G_{\delta,y}(s)=\frac{k_s+a}{2\cL_u s}
 =\underbrace{\frac{a}{\cL_u s}}_{\text{neutral pole}}
 +\underbrace{\frac{1}{2D_u\cL_u(k_s+a)}}_{\widehat G_\perp(s)},
 \label{eq:eft-moving-match}
\end{equation}
with $\widehat G_\perp/\widehat G_f=k_s/(k_s+a)$. Thus $\widehat G_\perp\sim\widehat G_f$ at high frequencies, but their zero-frequency ratio is $1/2$. The point mobility $a/\cL_u=2\kappa_*$ differs from that of the aligned source in \ref{app:collective-eft}. A more specific candidate, with feedback acting only on the point-driven phase and leaving its transverse drive unchanged, would dress the pole in Eq.~\eqref{eq:eft-moving-match} by $[1+g_R K(s)]^{-1}$, retaining the exact transverse remainder. This is a specified linear closure, not the general moving-root feedback law. In particular, the full root may require coupled shape--phase memory and induced forcing.

\begin{remark}[Conditional response matching]
The additive approximation in Eq.~\eqref{eq:eft-matched-candidate} requires an independent matching argument, with compatible forcing placements, mobility and regularisation, and control of shape--phase coupling. The exactly aligned construction has no transverse response to add.
\end{remark}

We now calculate the second cumulant of the candidate effective response. For the candidate common primary forcing, with stationary preparation of the stable filters and well-defined root increments, the existing spectral method gives
\begin{equation}
 \begin{aligned}
 \Xi_{\mathrm{EFT}}(\Delta)&=\frac{1}{2\pi\Delta}\int_{\mathbb R}
 4\sin^2(\omega\Delta/2)\\
 &\qquad\times|\widehat G_{\mathrm{EFT}}(\ii\omega)|^2 S_m(\omega)\,d\omega,\\
 |\widehat G_{\mathrm{EFT}}|^2&=|\widehat G_f|^2+|\widehat G_{\mathrm{eff}}|^2
 +2\operatorname{Re}(\widehat G_f\widehat G_{\mathrm{eff}}^*)
 \end{aligned}
 \label{eq:eft-matched-cumulant}
\end{equation}
Thus $\Xi_{\mathrm{EFT}}=\Xi_f+\Xi_c+\Xi_{\mathrm{cross}}$ for the sum in Eq.~\eqref{eq:eft-matched-candidate}; the same decomposition holds with matched transverse and collective transfers. Here the slope is already included in the transfers. Spatial orthogonality does not remove the cross term; multiple drives require their joint spectral matrix. The displacement itself need not be stationary. At balance in the ideal continuum, the bounded transverse transfer contributes subleading variance relative to the collective term, which has limiting rate $A_m C_\gamma(\kappa_*/g_R)^2$ for the aligned mobility, or the corresponding matched mobility for a different source. This does not demonstrate a finite common Abel and diffusive window. A positive reassessment-rate floor restores superdiffusive accumulation at sufficiently long lags in the scalar model. At finite resolution the physical cutoffs, cross-cumulant and root conditioning must be controlled.

\begin{remark}[Effective cumulant and projection]
The original kernel $\Xi(S,u;\Delta)$ is unchanged. The common-input cumulant retains the signed cross contribution and does not establish Markovian closure. Any replacement by $\Xi_{\mathrm{EFT}}$ requires a justified local closure and a separate adjoint-consistent clock and pricing-measure projection.
\end{remark}

\end{document}